\documentclass{article}

\usepackage{microtype}
\usepackage{graphicx}
\usepackage{subfigure}
\usepackage{booktabs} 

\usepackage{hyperref}

\usepackage[accepted]{icml2024}

\usepackage{amsmath}
\usepackage{amssymb}
\usepackage{mathtools}
\usepackage{amsthm}

\usepackage[capitalize,noabbrev]{cleveref}

\theoremstyle{plain}
\newtheorem{theorem}{Theorem}[section]

\theoremstyle{definition}
\newtheorem{definition}[theorem]{Definition}

\theoremstyle{remark}

\usepackage[textsize=tiny]{todonotes}

\usepackage{amsmath,bm}
\usepackage{amsthm}
\usepackage{caption}
\usepackage{subcaption}
\usepackage{mathtools}
\usepackage{multirow}
\usepackage{booktabs}
\usepackage{hhline}
\usepackage{wrapfig}
\usepackage{enumitem}

\DeclareMathOperator{\R}{\mathbb{R}}
\DeclareMathOperator{\E}{\mathbb{E}}
\newcommand{\bvec}[1]{\mathbf{#1}}

\newcommand\norm[1]{\left\lVert#1\right\rVert}

\DeclareMathOperator*{\argmax}{arg\,max}

\theoremstyle{definition}

\theoremstyle{definition}

\usepackage{comment}

\icmltitlerunning{E(3)-Equivariant Actor-Critic Methods for Cooperative Multi-Agent Reinforcement Learning}

\begin{document}

\twocolumn[
\icmltitle{${\rm E}(3)$-Equivariant Actor-Critic Methods for Cooperative Multi-Agent Reinforcement Learning}



\icmlsetsymbol{equal}{*}

\begin{icmlauthorlist}
\icmlauthor{Dingyang Chen}{yyy}
\icmlauthor{Qi Zhang}{yyy}
\end{icmlauthorlist}

\icmlaffiliation{yyy}{Artificial Intelligence Institute, University of South Carolina, Columbia, SC, USA}

\icmlcorrespondingauthor{Dingyang Chen}{dingyang@email.sc.edu}
\icmlcorrespondingauthor{Qi Zhang}{qz5@cse.sc.edu}

\icmlkeywords{Machine Learning, ICML}

\vskip 0.3in
]



\printAffiliationsAndNotice{} 

\begin{abstract}
Identification and analysis of symmetrical patterns in the natural world have led to significant discoveries across various scientific fields, such as the formulation of gravitational laws in physics and advancements in the study of chemical structures. In this paper, we focus on exploiting Euclidean symmetries inherent in certain cooperative multi-agent reinforcement learning (MARL) problems and prevalent in many applications. We begin by formally characterizing a subclass of Markov games with a general notion of symmetries that admits the existence of symmetric optimal values and policies. Motivated by these properties, we design neural network architectures with symmetric constraints embedded as an inductive bias for multi-agent actor-critic methods. This inductive bias results in superior performance in various cooperative MARL benchmarks and impressive generalization capabilities such as zero-shot learning and transfer learning in unseen scenarios with repeated symmetric patterns.
\end{abstract}

\section{Introduction}
\label{sec:Introduction}
It is widely believed by scientists that the organization and operation of our universe follow certain symmetry patterns and principles. These symmetry structures in physics lead to profound implications, such as the existence of conservation laws. 
When done properly, artificial intelligence (AI) can and has already benefited tremendously from exploiting these symmetries, with perhaps the most well-known example of convolutional neural networks (CNNs) being translation invariant to the input images \cite{Goodfellow-et-al-2016}.
Symmetries have also been identified and exploited for single-agent reinforcement learning (RL), where symmetric state-action pairs essentially define a homomorphism from the original Markov decision process (MDP) to a smaller abstract MDP \cite{ravindran2001symmetries}, which has been recently shown to be effective for deep single-agent RL \cite{van2020mdp,wang2022so,zhao2022integrating,mondal2022eqr,nguyen2023equivariant}. 

In this paper, we are interested in cooperative multi-agent reinforcement learning (MARL) problems, where symmetry structures are also prevalent and often exist in two forms.
The first form is permutation invariance which exists if agents are homogenous in terms of their effects on state dynamics and reward function, which has been well-studied in prior works, including both formalisms 
\cite{nguyen2017policy,yang2018mean,chen2022communicationefficient} 
and algorithmic techniques such as actor parameter sharing 
\cite{lowe2017multi,rashid2020monotonic,chen2022communicationefficient}, 
permutation-invariant centralized critics \cite{liu2020pic}, 
and mean-filed approximation \cite{yang2018mean}. 
We focus on the second form of Euclidean symmetries, including transformations of translation, rotation, and reflection, which exist if the MARL problem is situated in an Euclidean space.
Intuitively, MARL problems exhibit Euclidean symmetries whenever coordinate frames are used to provide references for the agents and their environment, because Euclidean transformations can be simply viewed as being applied to the reference frames, without changing the essence of the agents and environment.
We are particularly dealing with 3D Euclidean, i.e., ${\rm E}(3)$, multi-agent symmetries, which are prevalent in applications grounded in the physical world.

Although prevalent, multi-agent Euclidean symmetries are relatively underexplored, including only \citet{pol2022multiagent} and \citet{yu2023esp} as the prior work we know of.
Although sharing the same motivation, prior work falls short of providing ``genuine'' ${\rm E}(3)$ multi-agent symmetries in the sense that their Euclidean equivariance is only preserved for rotations of discrete angles in $\{k* 360^\circ/n\}_{k=0}^{n-1}$, which degenerates to the cyclic symmetries of ${\rm C}_n$, a subgroup of ${\rm E}(2)$ that is only 2D.
We identify challenges that have prevented prior work from exploiting continuous ${\rm E}(3)$ multi-agent symmetries and highlight our solutions as the three-fold contribution of this work:

\noindent\textbf{(i)} The first challenge is to find problem representations that are suitable to describe ${\rm E}(3)$-symmetries that are distributed among multiple agents.
We first rigorously formulate a subclass of Markov games (MGs) \cite{shapley1953stochastic}, group-symmetric MGs, that describes multi-agent symmetries by mathematical group transformations performed on states, actions, and agents' observations.
With ${\rm E}(3)$ being the special case, we use 3D point clouds to represent states, actions, and observations in a way that conveniently accommodates all transformations in ${\rm E}(3)$.

\noindent\textbf{(ii)} The second challenge is to develop architectures that are capable of exploiting continuous ${\rm E}(3)$-symmetries.
We prove several main properties of group-symmetric MGs, including symmetries in the value functions of symmetric policies and the existence of an optimal policy/value that is group-symmetric.
As a key difference from prior work, we exploit those properties by leveraging steerable message passing neural networks \cite{thomas2018tensor,e3nn_paper,brandstetter2022geometric} as the actor-critic architecture for MARL under the centralized training and decentralized execution (CTDE) paradigm, which are capable of preserving equivariance under all ${\rm E}(3)$ transformations. 

\noindent\textbf{(iii)} As a result, our method achieves superior sample efficiency and generalization performance in a range of benchmark MARL tasks that exhibit continuous ${\rm E}(3)$-symmetries but were not accommodated by prior work.

\section{Related Work}
\label{sec:Related Work}
In single-agent RL, the exploitation of symmetries has been originally formulated by \citet{ravindran2001symmetries} to reduce the redundancies in state and action spaces. Recent work mostly focuses on the combination of symmetries and deep RL. Some works \cite{laskin2020reinforcement,yarats2021image} utilize data augmentation, where an image-based observation undergoes transformation such as rotation and translation, to improve data efficiency. This approach mostly focuses on the invariant value functions with respect to the transformed input. 
Instead of generating more data through symmetric transformations which increases computation time, \citet{van2020mdp} build equivariant neural policies that directly support $C_n$ symmetries of discrete rotations. \citet{wang2022so2equivariant} extend $C_n$ symmetries to more general continuous $SO(2)$ symmetries that support equivariancy in translation and rotation in 2D. A recent work \cite{pmlr-v202-chen23i} considers a more general continuous ${\rm E}(3)$ symmetries in 3D space that additional supports equivariancy in reflection. 

Multi-agent symmetries are currently underexplored. Some work \cite{liu2020pic,chen2022communicationefficient} focus on the permutation invariance of homogeneous agents, and build permutation invariant value functions for better sample and computation efficiency. The work by \citet{li2021pairwise} considers symmetries specific to multi-agent pathfinding problems to reduce the search space. A recent work by \citet{pol2022multiagent} considers $C_n$ symmetries in MARL by extending the framework in the single-agent counterpart. However, it does not provide motivation as to why such symmetries result in equivariant actors and critics, and implementation-wise, their homomorphism network has an input size that scales with $|C_n|$. Instead, our work focuses on continuous ${\rm E}(3)$ multi-agent symmetries, and the implementation has the input size of the original observation dimension and still preserves equivariance for any continuous angle.

\citet{yu2023esp} share the same motivation of exploiting Euclidean symmetries for cooperative MARL. Our key differences lie in methodology and theoretical results.
Our work uses point cloud representations to characterize continuous ${\rm E}(3)$ symmetries, which facilitates the incorporation of ${\rm E}(3)$-equivariant neural networks, while \citet{yu2023esp} use traditional flat vector representations, which are not compatible with equivariant neural networks and they turn to data augmentation which is a less principled method.
Our theoretic results state the main properties of symmetric MGs. Besides the optimal value equivalence property included in \citet{yu2023esp}, we include theoretical results pertaining to observation-based policies, which are unique to multi-agent settings and provide theoretical justification for equivariant multi-agent actor-critic methods.

This work focuses on Euclidean symmetries commonly seen in practice, e.g., position and velocity in 3D space. Special neural architectures have been proposed to directly incorporate geometric inductive bias. Some works \cite{pmlr-v139-satorras21a,schutt2021equivariant,jing2020learning,le2022representation} preserve
equivariancy by equivariant operations in the original 3D Euclidean space. In contrast, other works lift the physical quantities from 3D space to higher-dimensional spaces for more expressive power, either through Lie algebra \cite{pmlr-v119-finzi20a} or spherical harmonics \cite{thomas2018tensor} and message passing \cite{brandstetter2022geometric}, which we adopt for actor-critic architectures. 

\section{Preliminaries}
\label{sec:Preliminaries}
\subsection{Cooperative Markov Games}
We use the framework of cooperative Markov game (MG) \cite{shapley1953stochastic} to formulate our cooperative multi-agent setting, which consists of
$N$ agents indexed by $i\in\mathcal{N}:=\{1,...,N\}$,
state space $\mathcal{S}$,
joint action space $\mathcal{A} = \mathcal{A}^1\times\cdots\times\mathcal{A}^N$ factored into local action spaces,
transition function $P: \mathcal{S}\times\mathcal{A}\to\Delta(\mathcal{S})$,
(team) reward function $r: \mathcal{S}\times\mathcal{A}\to\R$,
and initial state distribution $\mu \in \Delta(\mathcal{S})$,
where we use $\Delta(\mathcal{X})$ to denote the set of probability distributions over $\mathcal{X}$.
The MG evolves in discrete time steps:
at each time step $t$, all agents are in some state $s_t\in\mathcal{S}$ and each agent $i\in\mathcal{N}$ chooses its local action $a^i_t\in\mathcal{A}^i$, forming a joint action $a_t = (a^1_t,...,a^N_t)\in\mathcal{A}$ that induces a transition to a new state at the next time step according to the transition function, i.e., $s_{t+1}\sim P(\cdot | s_t, a_t)$, and the team reward signal $r_t:=r(s_t,a_t)$.
The state $s_0\sim\mu$ at time step $0$ is drawn from the initial state distribution.

\noindent\textbf{Full and partial observability.}
For ease of exposition, we will assume the agents can fully observe the states until Section \ref{sec:G-symmetric Markov games with Euclidean symmetries}, and our methods (Section \ref{sec:E(3)-equivariant multi-agent actor-critic methods}) and experiments (Section \ref{sec:Experiments}) accommodate the partial observability setting.
We shall consider a more general notion of full observability than directly observing the raw state:
each agent $i\in\mathcal{N}$ has access to a function $o^i: \mathcal{S}\to\mathcal{O}^i$ that maps the state space to its local observation space $\mathcal{O}^i$.
The agents effectively fully observe the state if and only if $o^i$ is bijective.
We will use the notion of $o=\{o^i\}_{i\in\mathcal{N}}$ and write $o(s) := \left(o^1(s),...,o^N(s)\right) \in \mathcal{O} :=  \times_{i\in\mathcal{N}}\Delta(\mathcal{O}^i)$.
The MG is therefore defined by tuple
$\langle \mathcal{N},\mathcal{S},\mathcal{A}, P, r, o \rangle$.

\noindent\textbf{Policies and values.}
The general, state-based joint policy, $\pi:\mathcal{S}\to\Delta(\mathcal{A})$, maps from the state space to distributions over the joint action space.
As the size of action space $\mathcal{A}$ grows exponentially with $N$, the commonly used joint policy subclass is the {\em product policy},  $\pi=(\pi^1,\cdots,\pi^N):\mathcal{S}\to\times_{i\in\mathcal{N}}\Delta(\mathcal{A}^i)$, which is factored as the product of local policies $\pi^i:\mathcal{S}\to\Delta(\mathcal{A}^i)$, $\pi(a|s) = \prod_{i\in\mathcal{N}}\pi^i(a^i|s)$, each mapping the state space only to the action space of an individual agent. 
Define the discounted return from time step $t$ as $R_t = \sum_{l=0}^{\infty} \gamma^l r_{t+l}$. For agent $i$, product policy $\pi=(\pi^1,...,\pi^N)$ induces a value function defined as $V_\pi(s_t) = \E_{s_{t+1:\infty}, a_{t:\infty}\sim\pi}[R_t|s_t]$, and action-value function $Q_\pi(s_t, a_t) = \E_{s_{t+1:\infty}, a_{t+1:\infty}\sim\pi}[R_t|s_t, a_t]$.
Following policy $\pi$, agent $i$'s cumulative reward starting from $s_0\sim\mu$ is denoted as $V_\pi(\mu):=\E_{s_0\sim\mu}[V_\pi(s_0)]$.
Under full observability, a state-based local policy can translate from and to an observation-based policy with $\nu^i(o^i(s))=\pi^i(s)$
and the corresponding observation-based product policy and its values are denoted as $\nu(o(s)) = (\nu^1(o^1(s)),...,\nu^N(o^N(s)))$, $Q_\nu$, $V_\nu$, respectively.

\subsection{Groups and Transformations}
Mathematical symmetry means a type of invariance: a property of an object remaining unchanged after some transformation.
This notion of symmetry can be formally described by invariant functions. A function $f:\mathcal{X}\to\mathcal{Y}$ is {\em invariant} under transformation operator $T: \mathcal{X}\to\mathcal{X}$ if $f\left(T[x]\right) = f(x)$ for any $x\in\mathcal{X}$.
More generally, function $f:\mathcal{X}\to\mathcal{Y}$ is {\em equivariant} under transformation operators $T: \mathcal{X}\to\mathcal{X}$ and $T': \mathcal{Y}\to\mathcal{Y}$ if $f\left(T[x]\right) = T'\left[ f(x)\right]$ for any $x\in\mathcal{X}$.
Most often, a symmetric object is invariant to not only one transformation but a set of transformations.
Such a set of symmetry transformations often forms a (mathematical) group $G$, which is a set of elements equipped with a binary operator satisfying the group axioms of closure, associativity, identity, and inverse \cite{dummit1991abstract}.
For function $f$ to be equivariant to group $G$, there exists transformation operators $T_g$ and $T'_g$, called {\em group actions}, associated with each group element $g\in G$, such that $f\left(T_g[x]\right) = T'_g\left[ f(x)\right]$ for any $x\in\mathcal{X}$.

\noindent\textbf{The Euclidean group ${\rm E}(3)$.}
We are mostly interested in group ${\rm E}(3)$ that comprises the group actions of translations, rotations, reflections, and finite combinations of them in 3D.
Therefore, it has as its subgroups the 3D translation group ${\rm T}(3)$, and the orthogonal group ${\rm O}(3)$ for 3D rotations and reflections.
The group actions of ${\rm O}(3)$ can implemented by multiplying square matrices, called {\em representation matrices}.
For example, for the group element of $g=(\alpha, \beta, \gamma)$ that rotates a 3D vector $x=\bvec{x}\in\R^3$ by $\alpha$, $\beta$, and $\gamma$, about $x$-, $y$-, and $z$- axes respectively, the rotation can be represented by matrix multiplication as $T_g[x] = \bvec{R}_{\alpha, \beta, \gamma}\bvec{x}$, where $\bvec{R}_{\alpha, \beta, \gamma}\in\R^{3\times3}$ is the 3D rotation matrix.

\begin{figure*}[tb]
\begin{center}
\includegraphics[width=\textwidth]{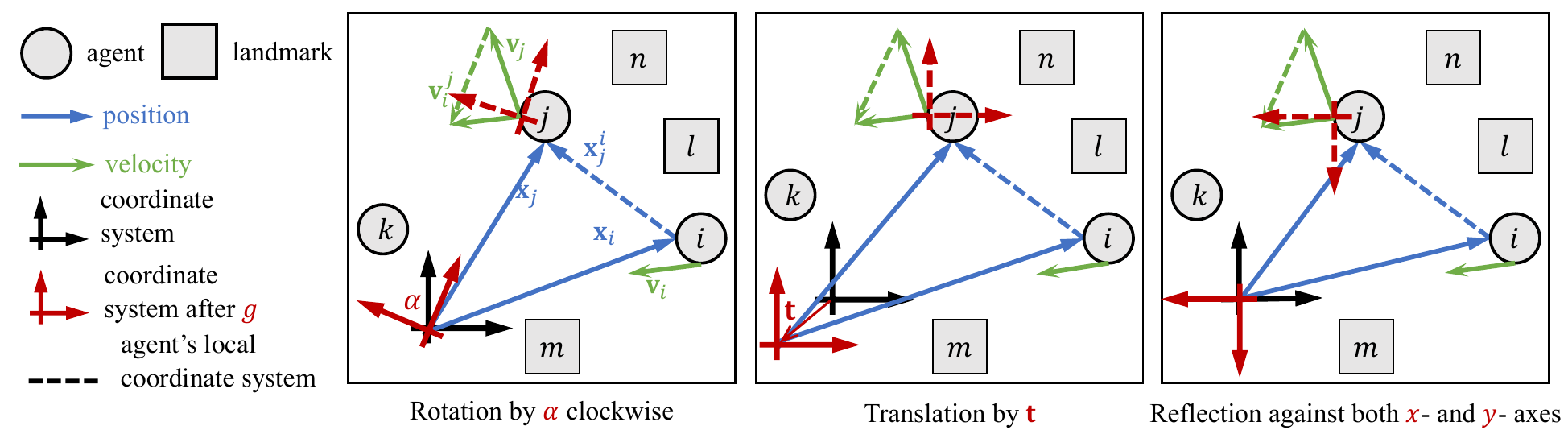}
\caption{Illustration of Cooperative Navigation ($N=3$) and its Euclidean symmetries.}
\label{fig:illustration}
\end{center}
\end{figure*}

\section{Markov Games with Euclidean Symmetries}
\label{sec:G-symmetric Markov games with Euclidean symmetries}
With the notions introduced, we are ready to define group-symmetric MGs by asking the transition, reward, and observations functions to be equivariant under group actions, as stated in Definition \ref{definition:G-symmetric MG}.
\begin{definition}[$G$-symmetric MG]
\label{definition:G-symmetric MG}
Consider MG $\langle \mathcal{N},\mathcal{S},\mathcal{A}, P, r, o \rangle$ and a (mathematical) group $G$ equipped with group actions $\{L_g: \mathcal{S} \to \mathcal{S}, K_g^s: \mathcal{A} \to \mathcal{A}, I_g^{s,a}: \R \to \R, H_g^s: \mathcal{O} \to \mathcal{O}\}_{g\in G}$.
The MG is {\em $G$-symmetric} if, for any $s,s'\in\mathcal{S}$, $a\in\mathcal{A}$, and $g \in G$, we have
\begin{align*}
    P(s'|s, a) =&~P\left(L_g[s'] ~|~ L_g[s], K_g^s[a]\right), \\
    r\left(L_g[s], K_g^s[a]\right) =&~I_g^{s,a} \left[r(s, a)\right], ~
    o\left(L_g[s]\right) =~H_g^s\left[o(s)\right]
\end{align*}
i.e., 
$P$, $r$, and $o$
are equivariant to the group actions.
\end{definition}

Definition \ref{definition:G-symmetric MG} is general enough to accommodate arbitrary symmetries in MGs.
For example, choosing the group actions to be permutations of $\mathcal{N}$ (i.e., $G={\rm S}(N)$) can describe permutation invariance between homogenous agents \cite{liu2020pic,chen2022communicationefficient}.

Moreover, these group actions in general are allowed to depend on the state and/or the action.
In this work, we focus on the special case of 3D Euclidean symmetries in MGs, i.e., $G={\rm E}(3)$, where the group actions are Euclidean transformations and do not depend on the state or the action.

For ease of exposition, we will ground our discussion with a running example of Cooperative Navigation, a popular cooperative multi-agent benchmark task \cite{lowe2017multi}.

\noindent\textbf{MPE's Cooperative Navigation and its symmetries.}
In Cooperative Navigation of Multi-Agent Particle Environment (MPE) \citep{lowe2017multi}, a popular MARL benchmark, $N$ agents move as a team to cover $N$ landmarks in a 2D space. The landmarks are randomly initialized at the beginning of an episode and fixed throughout the episode. The reward functions determine the reward for all agents according to the distances between all agents and the landmarks to encourage coverage, as well as penalties to discourage collisions if any. Under full observability, the observation of an agent contains its own absolute location and velocity, the relative locations and velocities of all other agents, and the relative locations of all the nearest landmarks.
Figure \ref{fig:illustration} (left) shows the case where $N=3$ and illustrates the rotational symmetry therein. The states and actions can be represented as 2D vectors under a global reference frame (e.g., agent $i$'s position, $\bvec{x}_i$), whereas each agent's observation is 2D vectors referenced to its local frame (e.g., agent $i$'s velocity relative to agent $j$, $\bvec{v}^j_i = \bvec{v}_i - \bvec{v}_j$).
Upon a rotation of all entities, which is shown in Figure \ref{fig:illustration} (left) as a rotation of the reference frames, the states, actions, and observations are rotated accordingly, yet the essence of the entities remain unchanged, and therefore the transition, reward, and observation functions are equivariant/invariant.
Similarly, it is easy to see that Cooperative Navigation also exhibits translational and reflectional symmetries, as shown in Figure \ref{fig:illustration} (middle and right, respectively).
Further, it is straightforward to extend these symmetries from 2D to 3D, thus making Cooperative Navigation ${\rm E}(3)$-symmetric. 

With the intuition from Cooperative Navigation, we now formally define ${\rm E}(3)$-symmetric MGs in Definition \ref{definition:E(3)-symmetric MGs}, where we will represent states and observations as 3D point clouds and specify the corresponding group actions applied to them.

\begin{definition}[${\rm E}(3)$-symmetric MGs]
\label{definition:E(3)-symmetric MGs}
MG $\langle \mathcal{N},\mathcal{S},\mathcal{A}, P, r, o \rangle$ is {\em ${\rm E}(3)$-symmetric} if the following conditions hold.
\begin{enumerate}[label=(\roman*),wide,itemsep=-.5pt]
\item \label{item:state_point_cloud}
A state is a 3D point cloud consisting of agents $\mathcal{N}$ and other uncontrollable entities $\mathcal{M}$, i.e., $s=\{(\bvec{x}_v,\bvec{f}_v)\}_{v\in\mathcal{V}}$, where 
    $\mathcal{V} = \mathcal{N} \cup \mathcal{M}$ is the set of all entities,
    $\bvec{x}_v\in\R^3$ is the 3D position vector of entity $v$ with feature vector $\bvec{f}_v\in\R^{d_v}$.

\item \label{item:action_space}
The local action spaces are Euclidean spaces, i.e., $\bvec{a}^i\in\mathcal{A}^i =\R^{d_{\bvec{a}^i}},\forall i\in\mathcal{N}$.

\item \label{item:observation_point_cloud}
Agent $i$'s observation is represented by a 3D point cloud consisting of the entities in its own view,
$o^i(s)=\{(\bvec{x}_v^i,\bvec{f}_v^i)\}_{v\in\mathcal{V}}$ where
$\bvec{x}_v^i\in\R^3$ is the 3D position vector of entity $v$ with feature vector $\bvec{f}_v^i\in\R^{d_{\bvec{f}_v^i}}$, both relative to  $i$.
\end{enumerate}

The MG's $P$ and $o$ are equivariant and $r$ is invariant
to the ${\rm E}(3)$ group actions in Definition \ref{definition:G-symmetric MG} that are specified as follows:

\begin{enumerate}[label=(\roman*), wide,itemsep=-.5pt]
\setcounter{enumi}{3}

\item \label{item:translation}
Under translation, the state is transformed with only the entities' positions being translated, i.e., 
$s=\{(\bvec{x}_v,\bvec{f}_v)\}_{v\in\mathcal{V}} 
~\to_g~ 
s=\{(\bvec{x}_v+\bvec{t}_g,\bvec{f}_v)\}_{v\in\mathcal{V}} 
$ where $\bvec{t}_g\in\R^3$ is the 3D translation for $g\in {\rm T}(3)$, 
while actions and observations remain unchanged.

\item  \label{item:rotation_reflection}
Under rotations or reflections,
all vectors in states, local actions, and observations are transformed according to their respective group representations.
For example, the state transforms as
$s=\{(\bvec{x}_v,\bvec{f}_v)\}_{v\in\mathcal{V}} 
~\to_g~ 
s=\{(\bvec{D}_g^{\bvec{x}} \bvec{x}_v,\bvec{D}_g^{\bvec{f}} \bvec{f}_v)\}_{v\in\mathcal{V}} 
$ where $\bvec{D}_g^{\bvec{x}}$ and $\bvec{D}_g^{\bvec{f}}$ are the representation matrices of $g\in {\rm O}(3)$ for the vector spaces of $\bvec{x}_v$ and $\bvec{f}_v$, respectively.
\end{enumerate}
\end{definition}

In essence, Definition \ref{definition:E(3)-symmetric MGs} requires the states and observations to be represented as 3D point clouds with the entities being the points with corresponding features, which are transformed according to Euclidean symmetries that are intuitively just changes of the reference frames.
The actions can be understood as features associated with the subset of agents that also are transformed accordingly.
In the example of Cooperative Navigation, 
the point feature vector for an entity $v\in\mathcal{V}$ includes its entity type to differentiate from agent vs landmark, which in our experiments is represented by a one-hot vector of two classes, $\bvec{f}_v\in\R^2$;
the action of agent $i\in\mathcal{N}$ includes its velocity, $\bvec{a}^i=\bvec{v}_i\in\mathcal{A}^i=\R^3$ with dummy $z$ values.

\noindent\textbf{Examples of ${\rm E}(3)$-symmetric MGs.}
Besides Cooperative Navigation, all other scenarios in MPE are ${\rm E}(3)$-symmetric.
Further, by Definition \ref{definition:E(3)-symmetric MGs}, any MG involving multiple entities interacting in a 2D/3D space is ${\rm E}(3)$-symmetric, which includes many real-world applications that are grounded in physical words, such as multi-robot systems, video games, materials design, etc.
Two other domains in our experiments that exhibit ${\rm E}(3)$-symmetric are the continuous control tasks in DeepMind Control Suite and game scenarios in StarCraft Multi-Agent Challenge, both being popular MARL benchmarks, with more details presented with our experiments. 
We give more examples of ${\rm E}(3)$-symmetric MGs in the appendix.

We would like to remark that,
in Definition \ref{definition:E(3)-symmetric MGs} that applies to the MARL benchmarks, an agent's observation uses a local reference frame with the agent’s position as the origin and with its orientation aligned with the global reference frame. So the agent’s own reference frame is the same as the global one up to a position shift, and that is the fundamental reason why observations, just like states, are transformed under rotations/reflections.
It is also straightforward to extend this definition (and the method) to the case where local reference frames have orientations different from the global reference frame, with several key points in doing so:
1) An agent’s local orientation should be part of its features and therefore part of the global state. Precisely, in Definition \ref{definition:E(3)-symmetric MGs}, the orientation should be included in feature vector $\mathbf{f}_v$, where $v$ is the vertex in the state Euclidean graph corresponding to the agent;
2) Upon a rotation, this orientation in feature vector $\mathbf{f}_v$ should be rotated accordingly;
and 3) If the observation is represented using the local reference frame, then it is often invariant to rotations, since the local reference frame is also rotated. 

We are now ready to derive the main properties of group symmetric MGs:
since the transition, reward, and observation functions are group-symmetric, if the policy is also group-symmetric, then we can expect its value functions are invariant under the group actions.
For example, under a rotation angle $\alpha$ if all agents in Cooperative Navigation always choose their velocities with the same rotation angle $\alpha$, then the policy value remains unchanged.
Definition \ref{definition:$G$-invariant MG policies} formally states this requirement in general $G$-symmetric MGs for both state-based and observation-based policies.

\begin{definition}[$G$-invariant MG policies]
\label{definition:$G$-invariant MG policies}
Let $\pi:\mathcal{S}\to\Delta(\mathcal{A})$ be a state-based policy in a $G$-symmetric MG.
We say $\pi$ is {\em $G$-invariant} if it is invariant to the group actions of $G$, i.e., for any $s\in\mathcal{S}$, $a\in\mathcal{A}$, and $g \in G$, 
    $
       \pi(a|s) =  \pi(K_g^s[a]~|~L_g[s])
    $.
Similarly, an observation-based policy $\nu:\mathcal{O}\to\times_{i\in\mathcal{N}}\Delta(\mathcal{A}^i)$ is {\em $G$-invariant} if
$
       \nu(a|o(s)) =  \nu(K_g^s[a]~|~H^s_g[o(s)]) 
$
for any $s\in\mathcal{S}$, $a\in\mathcal{A}$, and $g \in G$.
\end{definition}

Under the CTDE paradigm and bijective observation functions, the value functions of an observation-based product policy $\nu$ can directly condition on states, i.e., $V_\nu(s)$ and $Q_\nu(s,a)$.
We list the properties for observation-based $G$-invariant product policies and their state(-action) value functions in Theorem \ref{theorem:Main properties of $G$-symmetric MGs} below, which justifies how our method will exploit these properties under the CTDE paradigm. 

\begin{theorem}[Main properties of $G$-symmetric MGs, proof in the appendix]
\label{theorem:Main properties of $G$-symmetric MGs}
For a $G$-symmetric MG,
\begin{enumerate}[label=(\roman*),wide,itemsep=-.5pt]
\item \label{item:optimal G-invariant values}
The optimal values are $G$-invariant,
$
    V_*(s) =  V_*(L_g[s])
$, 
$
    Q_*(s,a) =  Q_*(L_g[s],K_g^s[a])
$.
\item \label{item:optimal G-invariant policy}
There exists an observation-based policy $\nu$ that are $G$-invariant and optimal,
$
       V_\nu(s) =  V_*(s), 
       Q_\nu(s,a) =  Q_*(s,a)
$.
\end{enumerate}
Further, for a $G$-invariant observation-based policy $\nu$,
\begin{enumerate}[label=(\roman*),wide,itemsep=-.5pt]\setcounter{enumi}{2}

\item \label{item:value equivalence}
Its value function is $G$-invariant: for any $s\in\mathcal{S}$, $a\in\mathcal{A}$, and $g \in G$, 
$
    V_\nu(s) =  V_\nu(L_g[s]),
    Q_\nu(s,a) =  Q_\nu(L_g[s],K_g^s[a])
$.

\item \label{item:G-invariant policy gradient}
Similarly, if $\nu$ is parameterized by $\theta$ as $\nu_\theta$ and differentialble, then
$
    \nabla_\theta \nu_\theta(a|o(s)) 
    =
    \nabla_\theta \nu_\theta(K^s_g[a]~|~o(L_g[s])) 
$.
\end{enumerate}
\end{theorem}
\noindent
Theorem \ref{theorem:Main properties of $G$-symmetric MGs} extends prior works on properties in single-agent symmetric MDPs \cite{ravindran2001symmetries,rezaei2022continuous}, and we are the first to take care of the distributed nature of the symmetries in MGs to make a rigorous proof. 
Theorem \ref{theorem:Main properties of $G$-symmetric MGs} establishes the properties our method will exploit next.
 
\section{\texorpdfstring{${\rm E}(3)$}{Lg}-Equivariant Multi-Agent Actor-Critic}
\label{sec:E(3)-equivariant multi-agent actor-critic methods}
The properties stated in Theorem \ref{theorem:Main properties of $G$-symmetric MGs} naturally prompt us with the idea of adopting group-invariant architectures for cooperative MARL.
In this work, we consider multi-agent actor-critic methods, such as MADDPG \cite{lowe2017multi} and MAPPO \cite{yu2022surprising}. 
Specifically, 
properties \ref{item:optimal G-invariant values} and \ref{item:optimal G-invariant policy} in Theorem \ref{theorem:Main properties of $G$-symmetric MGs} suggests that we can reduce the search of optimality within group-invariant functions for actors and critics.
Moreover, properties \ref{item:value equivalence} and \ref{item:G-invariant policy gradient} imply that group-invariant actors enjoy symmetric policy gradients:
\begin{align*}
    &\nabla_\theta \pi_\theta(a|s)  \cdot Q_{\pi_\theta}(s,a)
    \\=&
    \nabla_\theta \pi_\theta(K^s_g[a]~|~L_g[s])
    \cdot Q_{\pi_\theta}(L_g[s],K_g^s[a])
\end{align*}
which suggests that the optimization landscape is symmetric and therefore can be more efficiently optimized via gradient-based search \cite{zhao2022symmetry,zhao2023symmetries}.

\begin{figure}[b]
\begin{center}
\includegraphics[width=0.67\columnwidth]{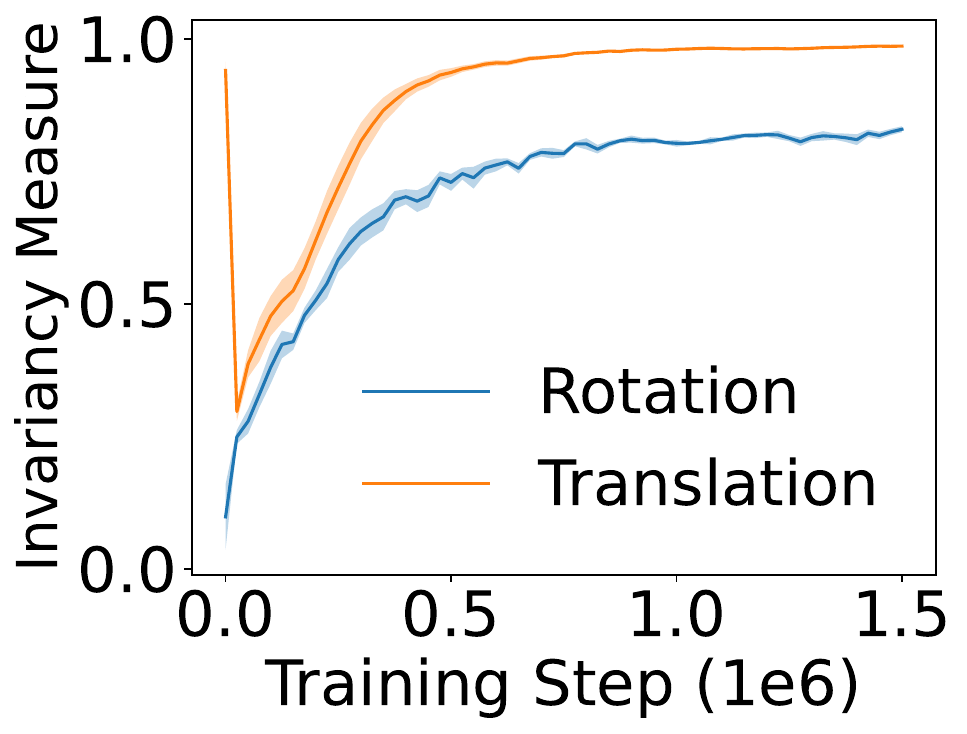}
\caption{Emergence of rotation- and translation-invariancy in MLP actors trained on 3-agent Cooperative Navigation.}
\label{fig:motivation_emergence}
\end{center}
\end{figure}

This idea is further motivated by our empirical analysis that finds emergence of group-invariancy in traditional actor-critic architectures that have no guaranteed invariancy.
Specifically, we train the agents in Cooperative Navigation ($N=3$) via MADDPG with MLP-based actors and critics.
Figure \ref{fig:motivation_emergence} plots the degree of the (observation-based) actors, $\{\nu^i\}_{i\in\mathcal{N}}$, being invariant to rotations and translations, respectively, during training.
For rotations, we select a finite set of angles, $A=\{30^\circ,60^\circ,\cdots,330^\circ\}$, and quantify the corresponding invariancy measure in state $s$ as 
\begin{align}
\label{eq:rotation_invariancy_measure}
\textstyle
    \frac{1}{|A|N}\sum_{i,\alpha\in A}
    \cos\left(
    {\rm rot}_\alpha[\nu^i(o^i(s)],~
    \nu^i(o^i({\rm rot}_\alpha[s]))
    \right)
\end{align}
where ${\rm rot}_\alpha[\cdot]$ performs the rotation by $\alpha$ and $\cos(\cdot,\cdot)$ measures the cosine similarity.
We measure the translation invariancy similarly.

\subsection{\texorpdfstring{${\rm E}(3)$}{Lg}-Equivariant Message Passing}
\label{subsec:E3-equivariant message passing}
In this work, we implement ${\rm E}(3)$-equivariant/invariant actor-critic architectures with ${\rm E}(3)$-equivariant message passing neural networks (E3-MPNNs) \cite{thomas2018tensor,e3nn_paper,brandstetter2022geometric}, a type of graph neural networks that process 3D {\em Euclidean graphs}, graphs where vertices are located at 3D positions, as an ${\rm E}(3)$-equivariant function.
Formally, an input Euclidean graph is represented as a tuple $G^{\rm in}=(\mathcal{V}, \mathcal{E}, \bvec{x}, \bvec{f}^{\rm in})$ where
$\mathcal{V}$ is a set of vertices with 3D positions $\bvec{x}_v\in\R^3$, 
$\mathcal{E}\subseteq \mathcal{V} \times \mathcal{V}$ is a set of edges,
and $\bvec{f}^{\rm in}$ is a set of feature vectors, each associated with a vertex or an edge.
The backbone of E3-MPNNs is ${\rm E}(3)$-equivariant message passing layers, denoted as $\texttt{E3-MP}(\cdot)$, maps the input Euclidean graph to an output Euclidean graph by updating only the feature set, $G^{\rm out}=(\mathcal{V}, \mathcal{E}, \bvec{x}, \bvec{f}^{\rm out})$, in an ${\rm E}(3)$-equivariant manner,
$
    \texttt{E3-MP}(T_g^{\rm in}[G^{\rm in}]) = T_g^{\rm out}[\texttt{E3-MP}(G^{\rm in})]
$
for $g\in{\rm E}(3)$,
where $T_g^{\rm in}$ and $T_g^{\rm out}$ include Euclidean transformations (i.e., translation, rotation, etc.) applied to the vertices and features in the input and output graph, respectively.
After the message passing layers, E3-MPNNs produce the final output $y\in\mathcal{Y}$ using a graph readout layer, $\texttt{E3-Readout}(\cdot)$, which is also ${\rm E}(3)$-equivariant,
$
    \texttt{E3-Readout}(T_g^{\rm out}[G^{\rm out}]) = T_g^{\mathcal{Y}}[\texttt{E3-Readout}(G^{\rm out})]
$
for $g\in{\rm E}(3)$.
This ensures overall equivariancy from the input graph to the final output.

\subsection{Integration Into Multi-Agent Actor-Critic Methods}
\begin{figure}[b]
\begin{center}
\includegraphics[width=.9\columnwidth]{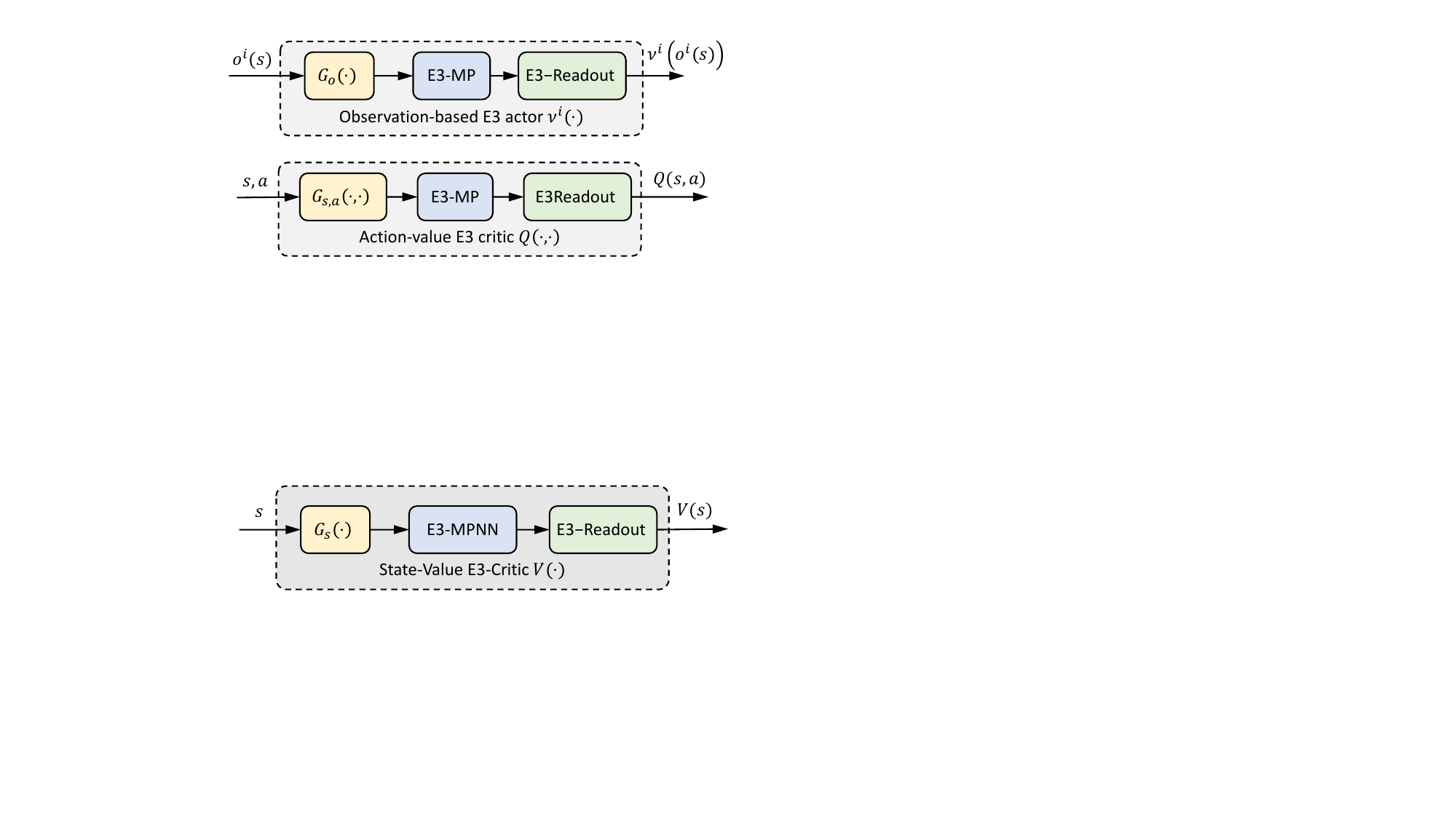}
\caption{Architectures for ${\rm E}(3)$-equivariant MADDPG.}
\label{fig:architecture}
\end{center}
\end{figure}

In order to use E3-MPNNs for multi-agent actor-critic architectures, we need to first represent state $s$, state-action pair $(s,a)$, and observation $o^i(s)$ as Euclidean graphs for centralized state-value critic $V(s)$, action-value critic $Q(s,a)$, and observation-based actor $\nu^i(o^i(s))$, respectively.
This is straightforward for ${\rm E}(3)$-symmetric MGs, because by Definition \ref{definition:E(3)-symmetric MGs} the states and observations are represented as 3D point clouds which can be directly cast into 3D Euclidean graphs if edges are added.
In our experiments, we use some heuristics to add edges, e.g., each vertex is connected with all others as a complete graph or with a fixed number of its nearest neighbors.
For a state-action pair, we simply treat the actions as additional feature vectors of individual agent vertices.
Figure \ref{fig:architecture} illustrates the ${\rm E}(3)$-equivariant action-value critic and observation-based actor for MADDPG.

We below highlight several implementation considerations. 
More details are provided in the appendix.

\noindent\textbf{SEGNN.}
For our actor-critic architecture, we employ steerable ${\rm E}(3)$ equivariant GNNs (SEGNNs), a recently developed E3-MPNN architecture that has superior performance on supervised learning tasks in computational physics and chemistry \cite{brandstetter2022geometric}.
SEGNNs split the input set of feature vectors into vertex feature vectors and edge feature vectors,  $\bvec{f}^{\rm in}=(\bvec{f}^{\rm node},\bvec{f}^{\rm edge})$, which comprises so-called {steerable} feature vectors concatenated by $(2l+1)$-dimensional vectors with $l=0,1,...$, often representing properties of the nodes (i.e., vertices) and edges.
For example, in Cooperative Navigation,  the node feature vector of agent $i$ is $\bvec{f}^{\rm node}_i=[\bvec{v}_i,\norm{\bvec{v}_i}, \bvec{a}_i,\norm{\bvec{a}_i},\rm node\_type]$, which is the concatenation of $(2l+1)$-dimensional vectors with $l\in\{0,1\}$.


\noindent\textbf{Observability.}
We adopt the paradigm of centralized training and decentralized execution, where the centralized critic can fully observe the state, while the local actors are observation-based with partial observability in general. For instance, in Cooperative Navigation, the velocities of other agents are not observable.

Our method and experiments are restricted to memory-less, non-recurrent actor-critic architectures even under partial observability. This is because realizing group-equivariancy in recurrent neural networks needs further technical treatments, which is left for future work.

To enforce uniform dimensionality required by SEGNN, we handle missing quantities for certain nodes due to partial observability by padding dummy vectors of zeros.

\section{Experiments}
\label{sec:Experiments}

\noindent\textbf{Environments.} 
We choose the popular cooperative MARL benchmarks of MPE, MuJoCo continuous control tasks (MuJoCo tasks), including the 2D ones from \citet{tassa2018deepmind} and 3D ones from \citet{pmlr-v202-chen23i} with single- and multi-agent variations, and StarCraft Multi-Agent Challenge (SMAC) \cite{samvelyan2019starcraft} to evaluate the effectiveness of our ${\rm E}(3)$-equivariant multi-agent actor-critic methods described in Section \ref{sec:E(3)-equivariant multi-agent actor-critic methods}.
Specifically, there are three task scenarios chosen in MPE, Cooperative Navigation, Cooperative Push, and Predator and Prey, where the collective goals of the controllable agents are to navigate the landmarks, push a ball to a target location, and catch all the preys, respectively. 
In MuJoCo tasks, we consider the representative tasks of cartpole, single- and multi-agent reacher, single- and multi-agent swimmer, multi-agent 3D hopper, and multi-agent 3D walker. The tasks vary in the degree of $\rm E(3)$ symmetries.
In SMAC, the selected scenarios are 8m\_vs\_9m with \textit{hard} difficulty and 6h\_vs\_8z with \textit{super hard} difficulty. 
These tasks exhibit different levels of ${\rm E}(3)$-symmetries.
All MPE and MuJoCo tasks exhibit perfect ${\rm E}(3)$-symmetries. 
SMAC, however, only exhibits imperfect ${\rm E}(3)$-symmetries due to two reasons:
1) the uncontrollable enemies in SMAC tasks might be ${\rm E}(3)$-symmetric, which might break the symmetries for the overall MG,
and 2) another reason is that the navigation actions are categorical in four Cartesian directions, which limit the actions' expressiveness for ${\rm E}(3)$-symmetries.

\noindent\textbf{Baselines.} 
In MPE, we choose MADDPG as the framework of the algorithms. We consider the classic implementation of both actor and critic by MLPs as the baseline. Another baseline \cite{liu2020pic} achieving state-of-the-art performance in MPE implements the critic by a graph convolutional neural network (GCN) that has no guarantee of being ${\rm E}(3)$-invariant. Our algorithms incorporate ${\rm E}(3)$-invariancy in critic and/or actor by SEGNN-based implementations. We denote the algorithms in the format of [critic\_type, actor\_type], with the baselines being [MLP, MLP] and [GCN, MLP], and our algorithms [SEGNN, SEGNN] and [SEGNN, MLP].
In MuJoCo tasks, we also apply MADDPG as the underlying framework of our algorithms. Similar to the case in MPE, the baselines are denoted as [MLP, MLP] and [GCN, MLP], and ours is denoted as [SEGNN, SEGNN]. For fair comparison, the input graph for the GCN-based critic is the same as that of the SEGNN-based one. In SMAC, the framework of the algorithms is MAPPO, where both actor and critic are commonly implemented by recurrent neural networks (RNNs) to incorporate trajectory-level information. However, group-equivariancy in RNNs needs further technical treatments so here we focus on non-recurrent MLP-based baselines, denoted as MAPPO-[MLP, MLP]. Another popular baseline we consider is QMIX \cite{rashid2020weighted} with recurrency, i.e., the value function is implemented by an RNN. Due to the categorical nature of SMAC's action space, our algorithm only incorporates ${\rm E}(3)$-invariancy in the critic, denoted as MAPPO-[SEGNN, MLP]. 

For fair comparison, all the algorithms have comparable amounts of parameters than those of the baselines, with the details in Appendix \ref{sec:Number of parameters in neural networks}.
Our code is publicly available at
\url{https://github.com/dchen48/E3AC}.

\noindent\textbf{Results overview.}
The MPE results show that the SEGNN-based critic and/or actor outperforms the baselines by a significant margin. Further, since the SEGNN-based architecture can deal with point clouds with an arbitrary number of points, the capability of zero-shot learning and transfer learning is empirically verified in MPE. Moreover, the emergence of invariancy of the baselines is found in all scenarios of MPE. In the MuJoCo tasks, the SEGNN-based architecture performs similarly to the MLP-based baseline in cartpole, consistent with the fact that the symmetries therein are relatively sparse.
In all the other selected tasks with inherently heavier ${\rm E}(3)$-symmetries, SEGNN-based architectures achieve noticeable improvements over the baseline. In SMAC, the experiments show that even the environment only partially satisfies the requirements of ${\rm E}(3)$-symmetric MGs, our SEGNN-based architecture can still be better performed in 6h\_vs\_8z. However, in SMAC, there is no obvious emergence of invariancy in MLP-based architectures, consistent with its limited ${\rm E}(3)$-symmetries. 

\subsection{Results on MPE}
\label{Results on MPE}
\noindent\textbf{Performance.}
In MPE, Figure \ref{fig:Main result} shows the learning curves comparing our SEGNN-based architecture described in Section \ref{sec:E(3)-equivariant multi-agent actor-critic methods} against the baselines of MLP- and GCN-based architecture.
The result clearly illustrates the effectiveness of the SEGNN-based algorithm. Specifically, by only using the SEGNN-based critic, i.e., [SEGNN, MLP] shown in the green curve, the learned architecture has already outperformed the baselines [GCN, MLP] and [MLP, MLP], shown in the blue and orange curves, respectively, in all scenarios in MPE. There is further performance boosting in Push\_N3, Prey\_N3, Navigation\_N6, and Prey\_N6 if we also use SEGNN-based actor, i.e., [SEGNN, SEGNN] shown with the red curves.

\begin{figure}[t]
\begin{center}
\centerline{\includegraphics[width=\columnwidth]{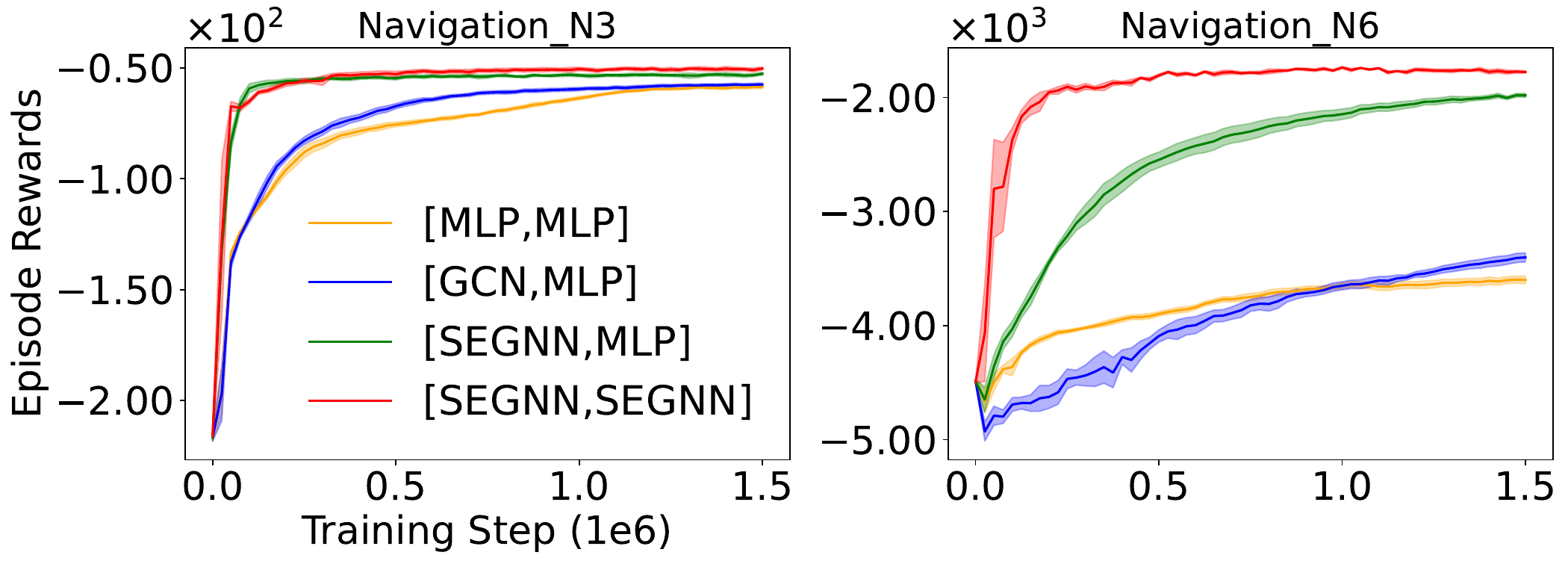}}
\centerline{\includegraphics[width=\columnwidth]{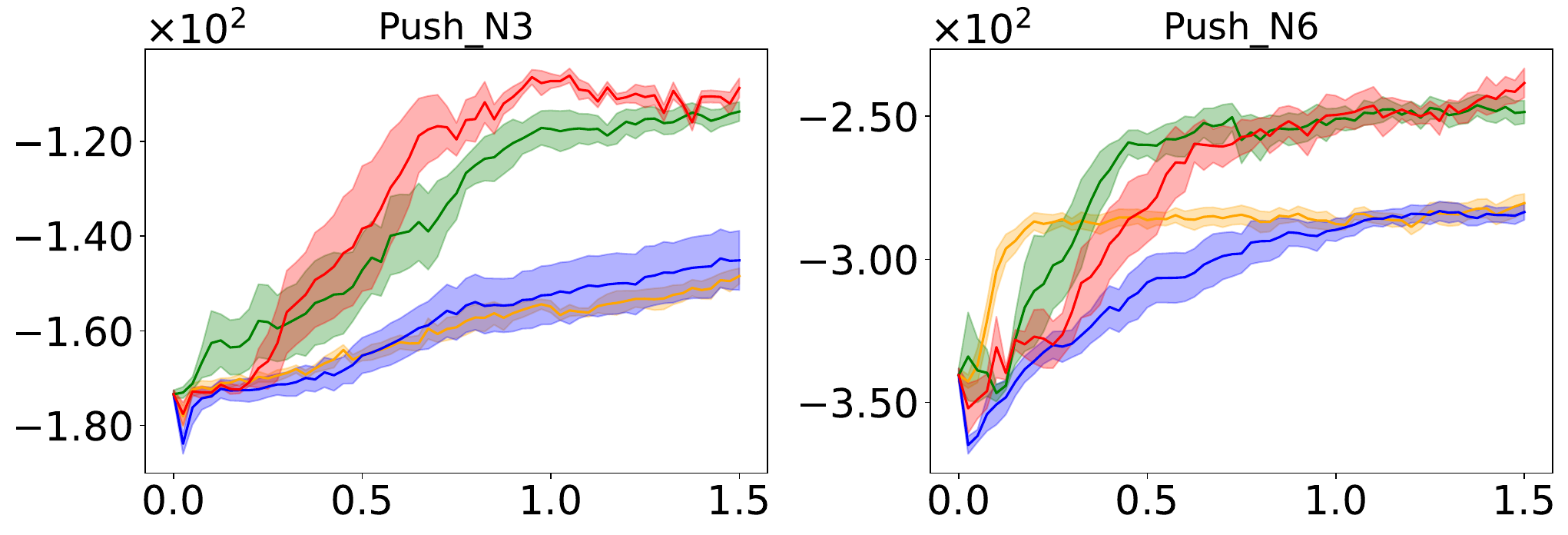}}
\centerline{\includegraphics[width=\columnwidth]{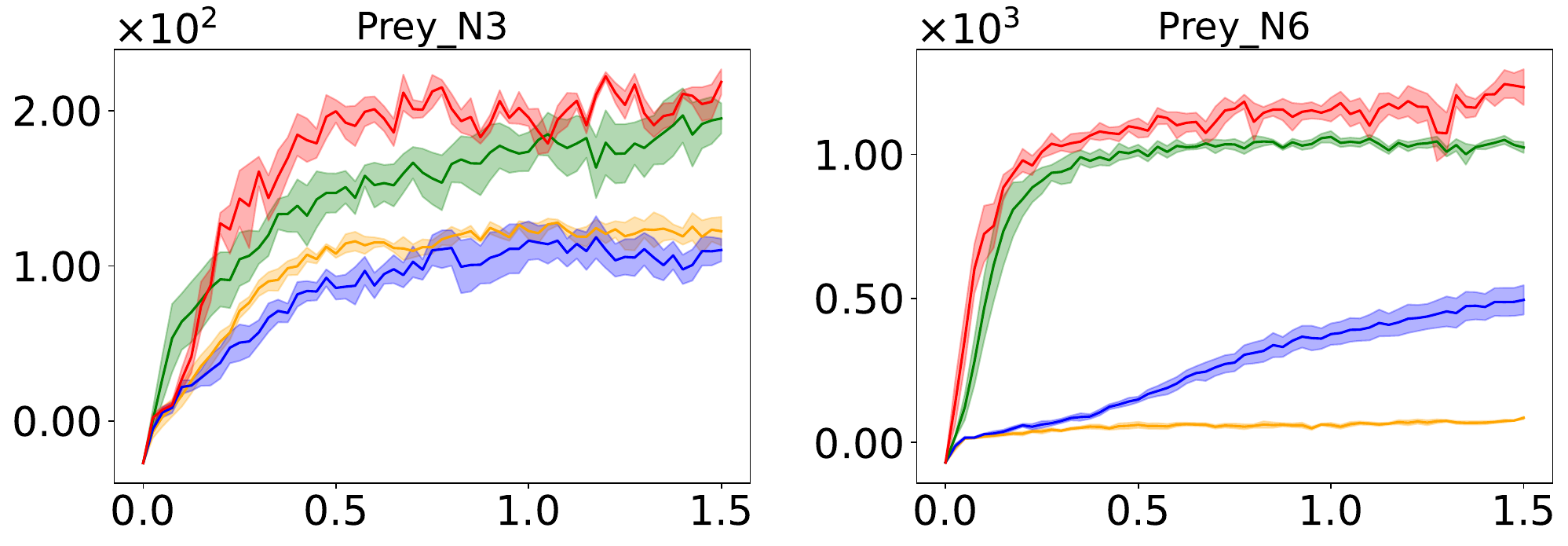}}
\caption{Performance comparison on MPE.}
\label{fig:Main result}
\end{center}
\vskip -0.3in
\end{figure}

\noindent\textbf{Zero-shot/transfer learning.}
\begin{figure}[b]
  \begin{center}
  \includegraphics[width=.7\columnwidth]{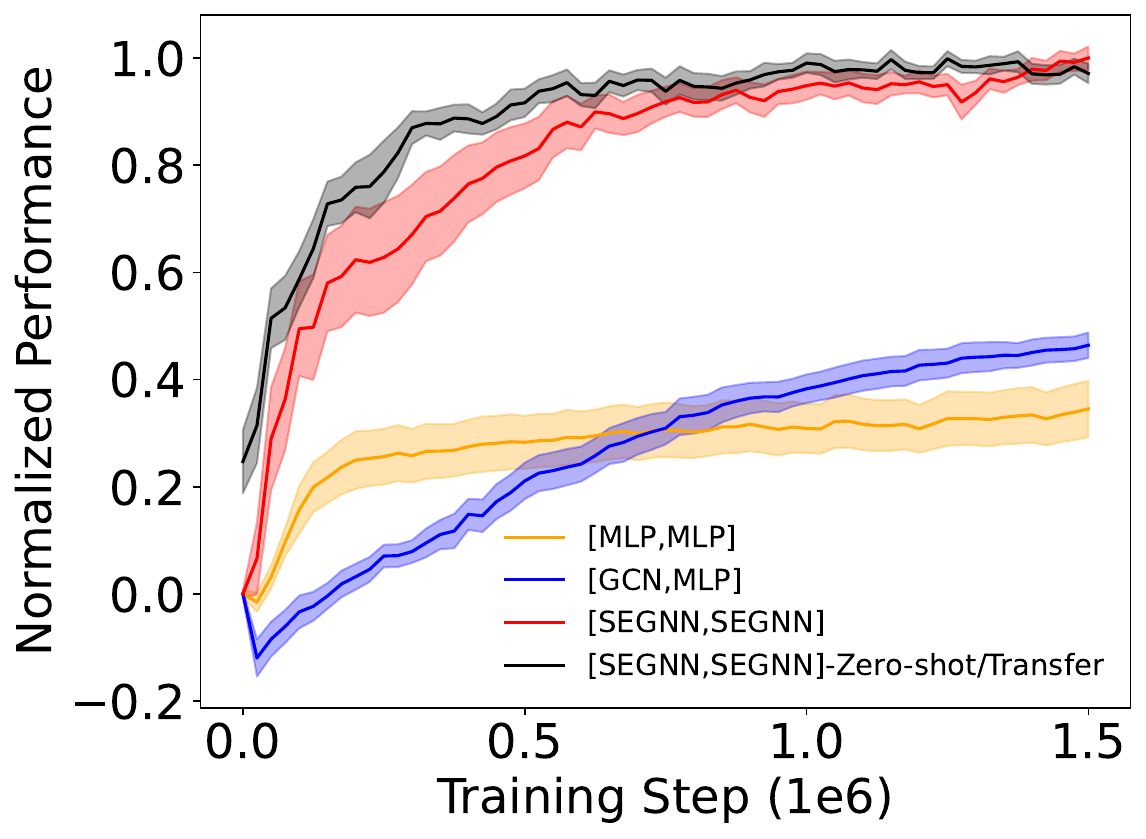}
  \end{center}
  \caption{Performance of zero-shot and transfer learning.}
\label{fig:zero-shot and transfer learning}
\end{figure}
Since SEGNN-based architectures can deal with point clouds of variable numbers of points, it is a natural question if they can perform well in different scenarios with similar setups, i.e., zero-shot learning. 
To answer this question, we test the performance of the SEGNN-based actors learned in 3-agent tasks for the corresponding 6-agent tasks.
The performance is normalized by linear mapping according to the performance of [SEGNN, SEGNN], whose initial performance is set to 0, and the final performance is set to 1. The result is averaged across all seeds and scenarios, plotted in Figure \ref{fig:zero-shot and transfer learning}. At step $0$, the performance of SEGNN-based architecture learned in the 3 agents' scenarios (black curve) is close to the final performance of [MLP, MLP] and that of [GCN, MLP] after training 0.75e6 steps. This empirically proves its good capability of zero-shot learning. Further, we can continue training the SEGNN-based architecture learned in 3 agents' scenarios in the corresponding 6 agents' scenarios to show its capability of transfer learning. Its effectiveness of transfer learning is empirically verified by the transferred architecture's faster convergence rate than the non-transferred counterpart. 

\noindent\textbf{Emergence of Euclidean invariancy.}
The GCN- and MLP-based architectures are by construction neither invariant to rotations nor translations with arbitrary weights. However, they may gradually adapt to have such abilities after training with a large number of data with invariant transformation, as shown in some previous works on data augmentation \cite{perez2017effectiveness} and contrastive learning \cite{chen2020simple}. 
The measurements of Euclidean invariancy are defined similarly to Equation (\ref{eq:rotation_invariancy_measure}), with details in Appendix \ref{sec:invariancy measurements in MPE} for MPE and \ref{sec:invariancy measurements in SMAC} for SMAC.



As illustrated in Figure \ref{fig:MPE-Full-Emergence of equivariancy} in the appendix, the actor trained with a SEGNN-based critic demonstrates the highest increase of rotation- and translation-invariancy, followed by those trained with GCN- and MLP-based critics. In contrast, there is no improvement in rotation- and translation-invariancy for both GCN- and MLP-based critics. This underscores the necessity of a SEGNN-based critic with inherent invariancy.

\subsection{Results on (Multi-Agent) MuJoCo Tasks}
\label{Results on MuJoCo}

\begin{figure}[t]
\begin{center}
\centerline{\includegraphics[width=1\columnwidth]{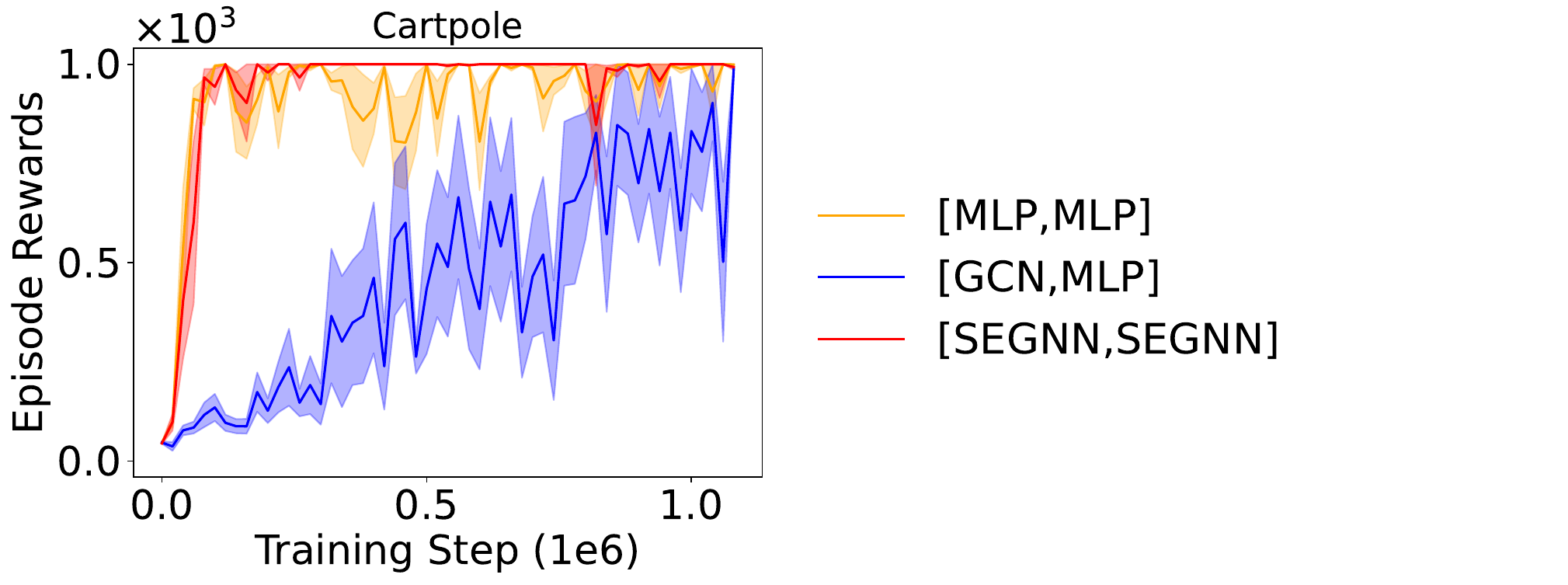}}
\centerline{\includegraphics[width=1\columnwidth]{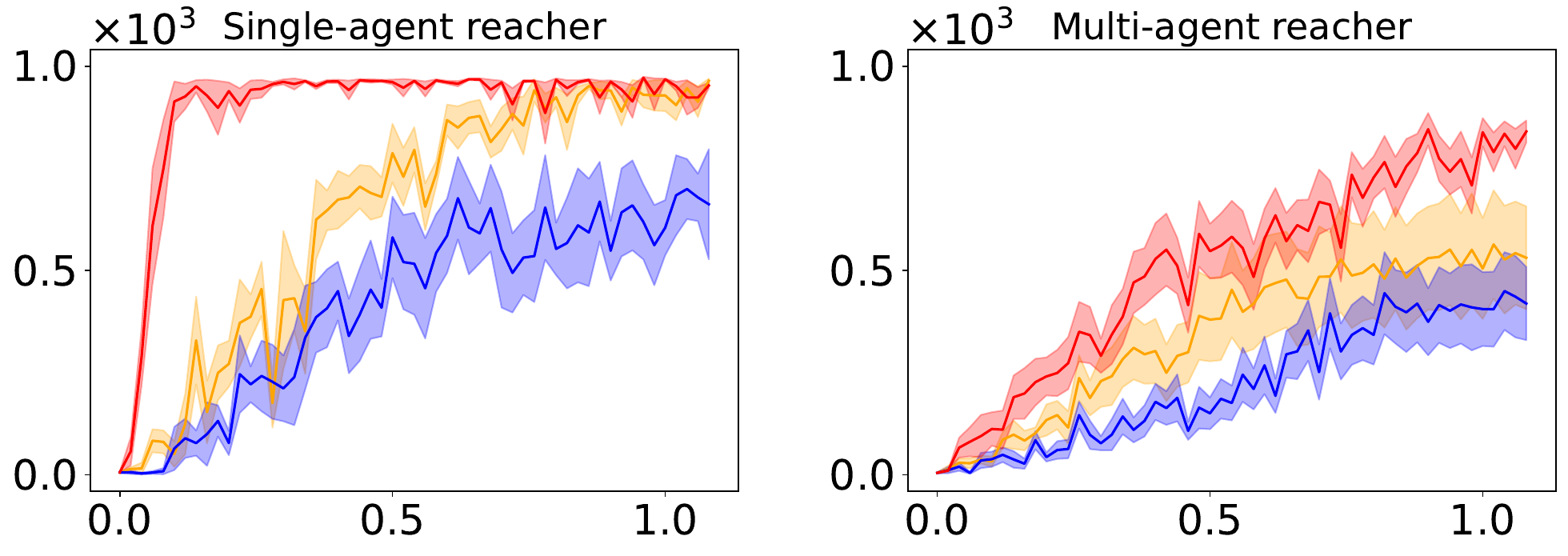}}
\centerline{\includegraphics[width=1\columnwidth]{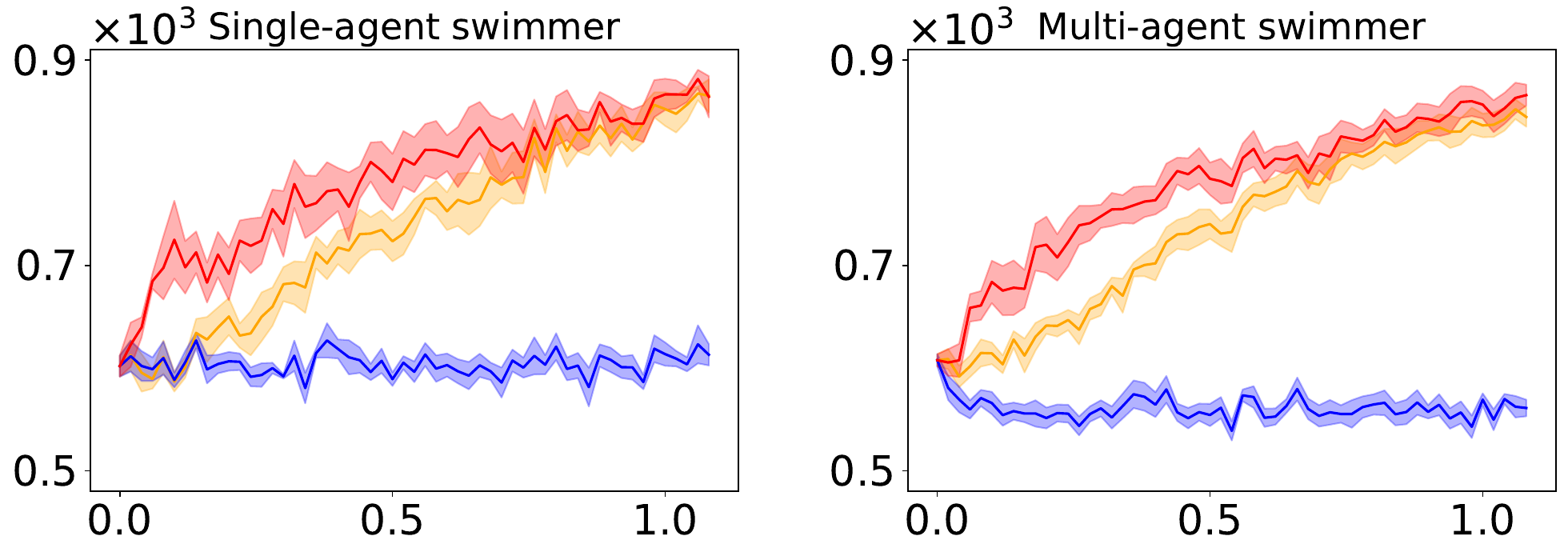}}
\centerline{\includegraphics[width=1\columnwidth]{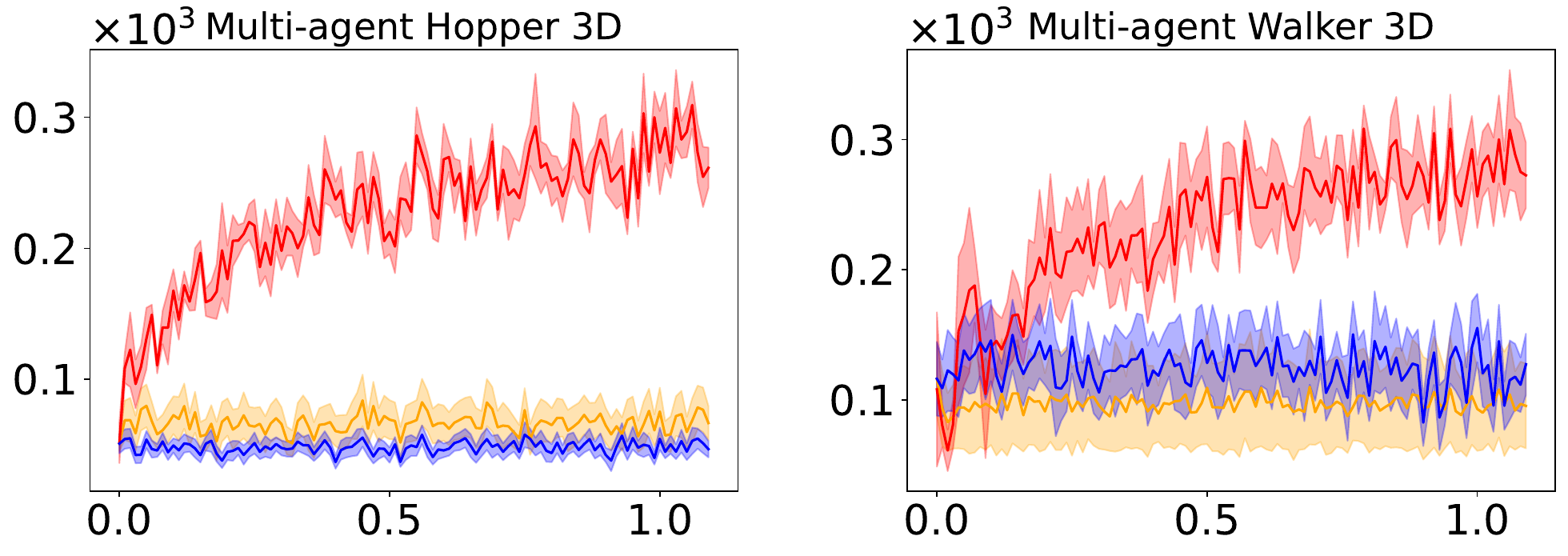}}
\caption{Performance comparison on MuJoCo tasks.}
\label{fig:MuJoCo_results_illustration}
\end{center}
\vskip -0.3in
\end{figure}

\noindent\textbf{Point cloud representations.} {MuJoCo tasks have Euclidean symmetries since they involve 3D robotic control. Here we choose several representative ones: cartpole (balance, sparse), single- and multi-agent reacher (hard), single- and multi-agent swimmer with two links, and 3D variants of multi-agent hopper and walker modified from the single-agent counterparts from \citet{pmlr-v202-chen23i}.
Note our method subsumes the single-agent setting as the special case of $N=1$:
we treat the single-agent tasks just like the multi-agent tasks, except that only one vertex in the point cloud has an action space to be controllable (because it’s single-agent), while all other vertices are non-controllable entities. 
This is contrastively different from standard methods for these single-agent tasks (i.e., [MLP, MLP] baseline) where all state variables are cluttered into a single vector.

The goals and the levels of ${\rm E}(3)$-symmetries of the selected scenarios are described in detail in Appendix \ref{sec:DMCS descriptions}.
The state contains physical attributes in $\R^3$, e.g., Cartesian coordinates, velocities, and angles of the robots' sub-components. The corresponding state-based point cloud in each task is described in detail in Appendix \ref{sec:DMCS point cloud details}.

\noindent\textbf{Performance.} The performance of the algorithms in MuJoCo tasks are shown in Figure \ref{fig:MuJoCo_results_illustration}. The selected tasks have different levels of symmetry, ranging from 2D (cartpole, reacher, and swimmer) to 3D (hopper and walker). The task cartpole (balance, sparse) only has rotation equivariancy of  $\rm{rot}_{180^\circ}$ with respect to the $z$ axis. Due to the lack of symmetries, we do not see much difference for SEGNN-based architectures against the baseline of [MLP, MLP]. The other 2D tasks, single- and multi-agent reacher, and single- and multi-agent swimmer inherently have stronger 2D rotation and translation equivariancy. Compared to the baseline of [MLP, MLP], our method [SEGNN, SEGNN] shows both faster convergence rate and better performance in single- and multi-agent reacher, and a faster convergence rate in single- and multi-agent swimmer. Due to gravity, which always points in the negative direction of the $z$ axis, the 3D tasks, multi-agent hopper and walker, also have 2D rotation and translation equivariancy in the $xy$ plane which is utilized by our method [SEGNN, SEGNN] to have superior performance compared to the baseline of [MLP, MLP]. In all the selected tasks, the baseline of [GCN, MLP] performs the worst, even if it uses the same input graph as the [SEGNN, SEGNN]. This illustrates the benefit of exploiting the inductive bias of Euclidean symmetries to reduce the search space of the neural networks' parameters.

\subsection{Results on SMAC}
\label{Results on SMAC}

\begin{figure}[!ht]
\begin{center}
\includegraphics[width=\columnwidth]{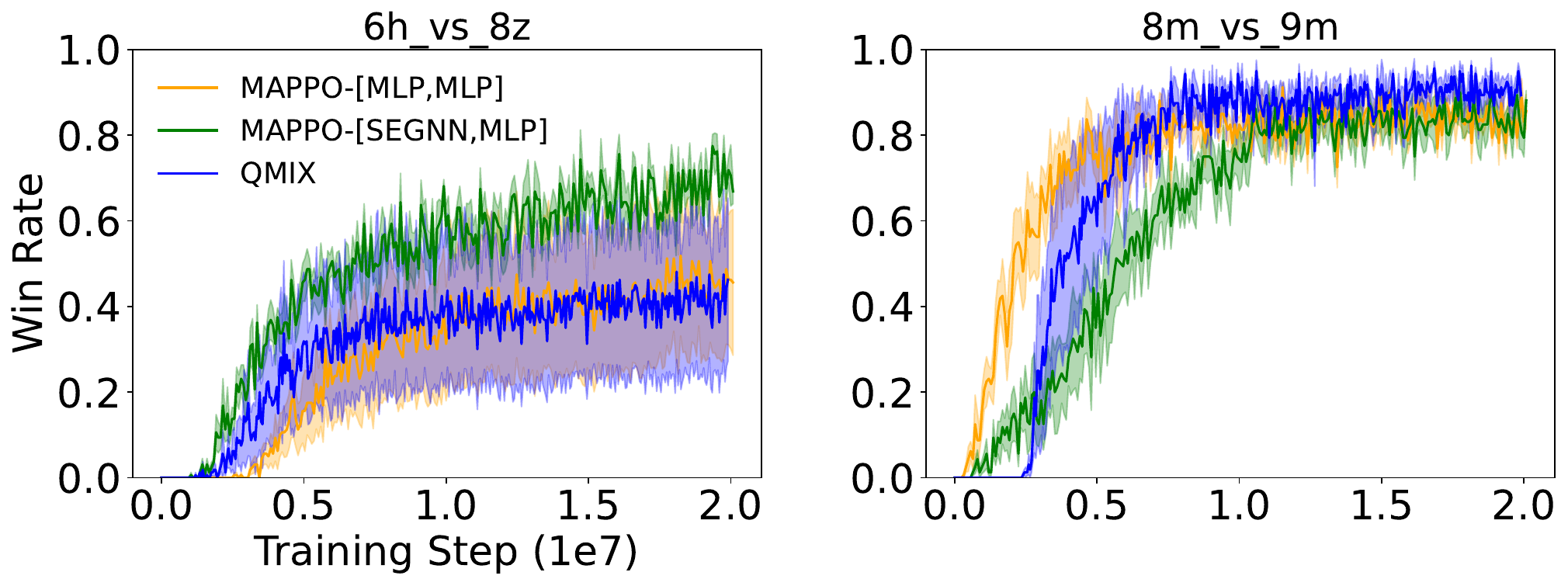}
\caption{Performance comparison on SMAC. 
}
\label{fig:Smac_main_performance}
\end{center}
\end{figure}

\noindent\textbf{Performance.}
The scenarios in SMAC \cite{samvelyan2019starcraft} are competitive in nature, so the invariancy of the game breaks due to the non-invariant behaviors of the uncontrollable enemies. Another possible reason for breaking invariancy is that the categorical actions in SMAC are not Euclidean spaces, as required by ${\rm E}(3)$-symmetric Markov game. Therefore, we may not observe the big increase in performance as the ones shown in MPE, especially in the non-competing ones navigation and push. We indeed see a mixed of performance in Figure \ref{fig:Smac_main_performance} that the SEGNN-based critic helps learn a better-performing actor in 6h\_vs\_8z but not in 8m\_vs\_9m.

\noindent\textbf{Emergence of Euclidean invariancy.}
As illustrated in Figure \ref{fig:SMAC-Full-Emergence of equivariancy} in the appendix, there is no obvious emergence of the rotation-invariancy for both actor and critic. This can be caused by the competitive nature of the scenarios, where we also see a decrease in actors' translation-invariancy in competitive scenarios Prey\_N3 and Prey\_N6. Another possible reason is the categorical nature of the actions, e.g., the move directions \{north, south, east, west\} are only rotation-equivariant to multiples of $90^\circ$ instead of arbitrary angles.

\section{Conclusion}
\label{sec:Conclusion}
We have developed formalisms and methods for exploiting symmetric structures in cooperative multi-agent reinforcement learning, with a focus on 3D Euclidean symmetries.
Specifically, we have formulated the novel notion of group-symmetric Markov games and derived its key properties that admit group-symmetric optimal values and policies. Consistent with these properties, we have discovered the emergence of Euclidean symmetries in vanilla MLP-based multi-agent actor-critic architectures. 
Then, we have developed ${\rm E}(3)$-equivariant message passing actor-critic architectures that specifically suit group-symmetric Markov games, which results in superior sample efficiency and generalization capabilities in most benchmark MARL tasks.

We observe two limitations of this work.
First, the proposed method requires knowledge and annotation of strict symmetries inherent in the multi-agent task, which might not be easily available. This prompts future work of automatic discovery of symmetries, which can even be approximate.
Second, we are restricted to memory-less, non-recurrent architectures in this work even under partial observability.

\section*{Acknowledgements}
Dingyang Chen acknowledges funding support from NSF award IIS-2154904.
Qi Zhang acknowledges funding support from NSF award IIS-2154904 and NSF CAREER award 2237963. Any opinions, findings, conclusions, or recommendations expressed here are those of the authors and do not necessarily reflect the views of the sponsors.
The authors thank the anonymous reviewers for their insightful and constructive reviews.

\section*{Impact Statement}

This paper presents work whose goal is to make the learning process of multi-agent reinforcement learning more efficient. There are many potential societal consequences of our work, none of which we feel must be specifically highlighted here.

\bibliography{example_paper}
\bibliographystyle{icml2024}

\newpage
\appendix
\onecolumn
\section{Proof of Theorem \ref{theorem:Main properties of $G$-symmetric MGs}}
\label{sec:Proof of Theorem theorem:Main properties of $G$-symmetric MGs}

We first prove property \ref{item:G-invariant policy gradient}  in Section \ref{sec:Proof of property item:G-invariant policy gradient} and then properties \ref{item:optimal G-invariant values}, \ref{item:optimal G-invariant policy}, and \ref{item:value equivalence} in Section \ref{sec:Proof of properties}.

\subsection{Proof of Property \ref{item:G-invariant policy gradient}}
\label{sec:Proof of property item:G-invariant policy gradient}
This property can be directly derived by the definition of $G$-invariant MG policies.
For state-based $G$-invariant policy $\pi$:
\begin{align*}
    \nabla_\theta \pi_\theta(a|s)
    = \lim_{\Delta\theta \to 0} \frac{\pi_{\theta+\Delta\theta}(a|s)-\pi_\theta(a|s)}{\Delta}
    = \lim_{\Delta\theta \to 0} \frac{\pi_{\theta+\Delta\theta}(K^s_g[a]~|~L_g[s])-\pi_\theta(K^s_g[a]~|~L_g[s])}{\Delta}
    = \nabla_\theta \pi_\theta(K^s_g[a]~|~L_g[s])
\end{align*}
where the second equality holds because we assume $\pi_\theta$ is $G$-invariant for any $\theta$.
The similar statement holds for observation-based $G$-invariant policy $\nu$.

\subsection{Proof of Properties \ref{item:optimal G-invariant values}, \ref{item:optimal G-invariant policy}, and \ref{item:value equivalence}}
\label{sec:Proof of properties}
These properties can be established by extending the notion of single-agent MDP homomorphism and its properties to MGs.

\begin{definition}[MG homomorphism]
\label{definition:MG homomorphism}
An MG homomorphism $(l, k_s, h_s)$ is a surjective map from an MG $\langle \mathcal{N},\mathcal{S},\mathcal{A}, P, r, o \rangle$ onto an abstract MG $\langle \mathcal{N},\overline{\mathcal{S}},\overline{\mathcal{A}}, \overline{P}, \overline{r}, \overline{o} \rangle$ by surjective maps $l: \mathcal{S}\to\overline{\mathcal{S}}$, $k_s: \mathcal{A}\to\overline{\mathcal{A}} $, and $h_s: \mathcal{O}\to\overline{\mathcal{O}}=\overline{\mathcal{O}}^1\times\cdots\times\overline{\mathcal{O}}^N$, such that
\begin{align}
    & r(s,a) = \overline{r}(l(s),k_s(a)) &\forall s\in\mathcal{S}, a\in\mathcal{A}\label{eq:MG_reward_invariance}\\
    & P([s']_l|s,a) = \overline{P}(l(s')|l(s),k_s(a))  &\forall s,s'\in\mathcal{S}, a\in\mathcal{A}\label{eq:MG_block_transition_equivariance}\\
    & o(l(s)) =  h_s(o(s)) &\forall s\in\mathcal{S}\label{eq:MG_observation_equivariance}.
\end{align}
\end{definition}

\begin{definition}[Policy lifting in MG homomorphisms]
\label{definition:Policy lifting in MG homomorphisms}
We define policy lifting for state-based and observation-based policies separately, assuming full observability: 
\begin{itemize}
    \item State-based MG policy lifting. 
    Let $\pi_\uparrow:\mathcal{S}\to\Delta(\mathcal{A})$ and $\overline{\pi}:\overline{\mathcal{S}}\to\Delta(\overline{\mathcal{A}})$ be two state-based policies in the actual and abstract MGs, respectively.
    We say $\pi_\uparrow$ is a lift of $\overline{\pi}$, denoted as $\pi_\uparrow=\texttt{lift}(\overline{\pi})$, if $\sum_{a\in k_s^{-1}(\overline{a})}\pi_\uparrow(a|s)=\overline{\pi}(\overline{a}|l(s))$ for any $s\in\mathcal{S}$ and $\overline{a}\in\overline{\mathcal{A}}$.
    
    \item Observation-based MG policy lifting.
    Let $\nu_\uparrow = \{\nu_\uparrow^i: \mathcal{O}^i\to\Delta(\mathcal{A}^i)\}_{i\in\mathcal{N}}$ and $\overline{\nu} = \{\overline{\nu}^i:\overline{\mathcal{O}}^i\to\Delta(\overline{\mathcal{A}}^i)\}_{i\in\mathcal{N}}$ be two observation-based policies in the actual and abstract MGs, respectively.
    We say $\nu_\uparrow$ is a lift of $\overline{\nu}$, denoted as $\nu_\uparrow=\texttt{lift}(\overline{\nu})$, if $\sum_{a\in k_s^{-1}(\overline{a})}\nu_\uparrow(a|o(s))=\overline{\nu}(\overline{a}|h_s(o(s)))$ for any $s\in\mathcal{S}$ and $\overline{a}\in\overline{\mathcal{A}}$.
\end{itemize}
\end{definition}

\begin{theorem}[Value equivalence in MG homomorphisms]
\label{theorem:Value equivalence in MG homomorphisms}
Consider an MG homomorphism $(l, k_s, h_s)$ from MG $M=\langle \mathcal{N},\mathcal{S},\mathcal{A}, P, r, o \rangle$ with bijective observation functions to abstract MG $\overline{M} = \langle \mathcal{N},\overline{\mathcal{S}},\overline{\mathcal{A}}, \overline{P}, \overline{r}, \overline{o} \rangle$.
We have 
\begin{align}
\label{eq:optimal value equivalence in MG homomorphisms}
    V_*(s) = \overline{V}_*(l(s)), \quad
    Q_*(s,a) = \overline{Q}_*(l(s), k_s(a))
    \quad\text{for any $s\in\mathcal{S}, a\in\mathcal{A}$.}
\end{align}
For a state-based abstract policy $\overline{\pi}$ with its lift $\pi_\uparrow$, we have 
\begin{align}
\label{eq:value equivalence under policy lifts in MG homomorphisms}
    V_{\pi_\uparrow}(s) = \overline{V}_{\overline{\pi}}(l(s)), \quad
    Q_{\pi_\uparrow}(s,a) = \overline{Q}_{\overline{\pi}}(l(s),k_s(a))
    \quad\text{for any $s\in\mathcal{S}, a\in\mathcal{A}$.}
\end{align}
For an observation-based abstract policy $\overline{\nu}$ with its lift $\nu_\uparrow$, we have 
\begin{align}
\label{eq:value equivalence under observation-based policy lifts in MG homomorphisms}
    V_{\nu_\uparrow}(o(s)) = \overline{V}_{\overline{\nu}}(h_s(o(s))), \quad
    Q_{\nu_\uparrow}(s,a) = \overline{Q}_{\overline{\nu}}(h_s(o(s)),k_s(a))
    \quad\text{for any $s\in\mathcal{S}, a\in\mathcal{A}$.}
\end{align}
\end{theorem}
\begin{proof} [Proof of Theorem \ref{theorem:Value equivalence in MG homomorphisms}]
For state-based policies, we use similar proof techniques for single-agent MDP homomorphisms \cite{ravindran2001symmetries,rezaei2022continuous}.
Under bijective observation functions, we can prove \eqref{eq:value equivalence under observation-based policy lifts in MG homomorphisms} by first translating them into state-based policies and then applying \eqref{eq:value equivalence under policy lifts in MG homomorphisms}.
\end{proof}

\noindent\textbf{Proof of properties \ref{item:optimal G-invariant values} and \ref{item:value equivalence}.}
We first establish that a $G$-symmetric MG induces a MG homomorphism, where $G$-invariant MG policy can be viewed as lifted from the induced abstract MG.
Then, properties \ref{item:optimal G-invariant values} and \ref{item:value equivalence} directly follow from Theorem \ref{theorem:Value equivalence in MG homomorphisms}.

Specifically, consider a $G$-symmetric MG as defined in Definition \ref{definition:G-symmetric MG}.
Define maps $(l,k_s,h_s)$ as follows, such that the state-action-observation tuples in the same orbit under $G$ map to the same abstraction:
\begin{align*}
    l(s) = l(L_g[s]), \quad
    k_s(a) = k_{L_g[s]}(K^s_g[a]), \quad
    h_s(o(s)) = h_{L_g[s]}(o(L_g[s])) = h_{L_g[s]}(H^s_g[o(s)])
\end{align*}
for any $s\in\mathcal{S}, a\in\mathcal{A}, g\in G$.
It is easy to verify that maps $(l,k_s,h_s)$ are surjective and satisfy conditions \eqref{eq:MG_reward_invariance}, \eqref{eq:MG_block_transition_equivariance}, and \eqref{eq:MG_observation_equivariance} in Definition \ref{definition:MG homomorphism}, and therefore maps $(l,k_s,h_s)$ induce an MG homomorphism.
Property \ref{item:optimal G-invariant values} then directly follows from \eqref{eq:optimal value equivalence in MG homomorphisms} and the assumption that observation functions are bijective.

Under this an MG homomorphism, a state-based $G$-invariant policy induces an abstract policy, $\pi(a|s) = \pi(K^s_g[a] ~|~ L_g[s]) = \overline{\pi}(\overline{a}|\overline{s})$, where $(\overline{s}, \overline{a}) = (l(s), k_s(a)) = (l(L_g[s]), k_{L_g[s]}(K^s_g[a]))$.
Because $K^s_g[\cdot]$ is reversible via $K^s_{g^{-1}}[\cdot]$, $k_s$ is actually bijective.
To see this, we can pick an arbitrary state-action pair in every state-action orbit as the canonical state-action pair for that orbit,
then the mapping of any state-action pair can be done via two steps, from the state-action pair to the canonical state-action pair and then to the abstraction, with both steps being bijective.
Thus, we have $\sum_{a'\in k_s^{-1}(\overline{a})}\pi(a'|s) = \pi(a|s) =\overline{\pi}(\overline{a}|\overline{s})$, and therefore $\pi$ is a lift policy from $\overline{\pi}$.
Property \ref{item:value equivalence} for $\pi$ then directly follows from \eqref{eq:value equivalence under policy lifts in MG homomorphisms}.
Under bijective observation functions, we can similarly derive property \ref{item:value equivalence} for state-based $G$-invariant policy $\nu$ from \eqref{eq:value equivalence under observation-based policy lifts in MG homomorphisms}.

\noindent\textbf{Proof of property \ref{item:optimal G-invariant policy}.}
Such a state-based policy $\pi$ can be induced from the optimal $G$-invariant value function.
Letting $Q_*$ be the optimal action-value function that is $G$-invariant, the existence of which has been proved as property \ref{item:optimal G-invariant values}, which induces a deterministic optimal state-based policy $\pi_*$ as the greedy policy with respect to $Q_*$, i.e.,  $\pi_*(s) = \argmax_a Q_*(s,a)$ for any $s\in\mathcal{S}$.
Since $Q_*$ is $G$-invariant, it is easy to see that $\pi_*$ is also $G$-invariant, $K^s_g[\pi_*(s)] = \pi_*(L_g[s])$.
Such an observation-based policy $\nu$ can be similarly derived. 

\section{Examples of Group-Symmetric MGs}
\label{sec:Examples of group-symmetric MGs}

\subsection{\texorpdfstring{${\rm E}(3)$}{Lg}-Symmetric MGs}
It is quite common that in the 3D physical world, entities' positions and some other physical quantities are described with respect to an arbitrary reference whose change has no impact on the behaviors of the entities. Such scenarios can be described as an ${\rm E}(3)$-symmetric MGs. Examples include traffic with vehicles, team sports, surveillance with drones, and games such as SMAC and MPE.

Below we highlight three MARL benchmark games to show the generality of ${\rm E}(3)$-symmetric MGs.

\noindent\textbf{Multi-Agent MuJoCo Tasks.}
In a Multi-Agent MuJoCo task, a given single robotic agent
is split into several sub-parts (agents). Each (sub-part) agent can be viewed as a point in the 3D point cloud, with the its centroid as the 3D position. The ${\rm E}(3)$-equivariant features can be velocities, forces, etc, and the ${\rm E}(3)$-intvariant features can be id, unit\_type, etc. The actions are the forces applied to each sub-parts. 

\noindent\textbf{Google Research Football (GRF)}
In GRF,  each unit in the field, such as a player and a ball can be viewed as a point in the 3D point cloud, with ${\rm E}(3)$-equivariant features such as the velocity and ${\rm E}(3)$-invariant features such as the unit\_type. The actions are categorical, with movement actions and some other non-movement actions such as shooting and passing. The non-movement actions are ${\rm E}(3)$-invariant. If the player can move to any direction, i.e., the movement actions are continuous, then the game satisfies the requirement of the action space in \ref{item:action_space}. The scenarios in GRF are competitive where the (uncontrollable) enemies are viewed as part of the environment. The game is an ${\rm E}(3)$-symmetric MGs if the enemies' policies are ${\rm E}(3)$-invariant.

\noindent\textbf{SMAC}
In SMAC, each unit in the combat can be view as a point in the 3D point cloud, with dummy $z$ coordinate. There are some competitive scenarios in SMAC, where the controllable agents are fighting against some (uncontrollable) enemies which are viewed as part of the environment. Therefore, similar to GRF, whether the game is an ${\rm E}(3)$-symmetric MGs or not depends on the ${\rm E}(3)$-invariancy of the uncontrollable enemies' policies.
The actions are also categorical including movement actions and some other ${\rm E}(3)$-invariant non-movement actions such as attacking.

\subsection{\texorpdfstring{${\rm S}_N$}{Lg}-Symmetric MGs}
    Besides the Euclidean transformation, the permutation is another common transformation in the physical world. Intuitively, for homogeneous entities with the same exact behaviors, permute them will have no impact. Such multi-agent symmetries can be captured by the symmetric group ${\rm S}_N$ with permutations as the group actions. The Homogeneous MGs, which is ${\rm S}_N$-symmetric MGs under our definition, is formallly defined by the work \cite{chen2022communicationefficient} with the permutation transformation defined on state, action, observation, and the transition and the reward function. Enforcing permutation-invariancy has been found beneficial by \cite{chen2022communicationefficient,liu2020pic}. Our definition of the ${\rm G}$-symmetric MGs is general enough to cover it. 

The examples include MPE tasks with homogeneous agents, SMAC scenarios with homogeneous ally units, team sports with homogeneous players (assume players have the same capabilities), traffic with homogeneous vehicles, and surveillance with homogeneous drones.

\section{Experiment Details}
\label{sec:Experiment details}
Below we describe in general notation of states and observations in scenarios in MPE and SMAC. 
\subsection{MPE}
Multi-Agent Particle Environment (MPE) with the efficient implementation by \cite{liu2020pic} is a classical benchmark with homogeneous agents for multi-agent reinforcement learning, each of which has  versions with $N=3,6$ agents, respectively.
These MPE environments can be cast as ${\rm E}(3)$-Symmetric MGs, where we have dummy $z$ coordinate. Below we describe the general form of state and observation in all scenarios of MPE, and the corresponding state-action based point cloud and observation-based point cloud which are processed by SEGNN-based critic and actor, respectively, and also the details of each scenarios in MPE. 

\textbf{Notation} We denote the set of agents, landmarks, and preys, which are basic units in scenarios of MPE, as $\mathcal{N}^a=\{1,...,N^a\}$, $\mathcal{N}^l=\{N^a+1,...,N^l\}$, and $\mathcal{N}^p=\{N^a+N^l+1,...,N+N^l+N^p\}$, where ${N^a},{N^l}$, and ${N^p}$ are the number of agents, landmarks and preys, respectively. Here only the agents are controllable so $N^a=N$. 
We have the set of nodes $\mathcal{V}=\mathcal{N}^a\cup\mathcal{N}^l\cup\mathcal{N}^p$.
The absolute positions associated with $\mathcal{N}^a,\mathcal{N}^l$ and $\mathcal{N}^p$ are $\{\bvec{x}_i\}_{i\in \mathcal{N}^a},\{\bvec{x}_v\}_{v\in \mathcal{N}^l}$ and $\{\bvec{x}_v\}_{v\in \mathcal{N}^p}$, respectively. 
The absolute velocities associated with $\mathcal{N}^a$ and $\mathcal{N}^p$ are $V^{a}=\{\bvec{v}_i\}_{i\in \mathcal{N}^a}$ and $V^{p}=\{\bvec{v}_v\}_{v\in \mathcal{N}^p}$, respectively. The landmarks are fixed, so the associate velocities is $V^{l}=\{\bvec{0}\}_{i\in \mathcal{N}^l}$. Let $V=V^a\cup V^l\cup V^p$, which is the set of velocities associated with the set of nodes $\mathcal{V}$.

\textbf{State} The state consists of the positions $\bvec{x}$ of the set of entities $\mathcal{N}$, with the associated features $\bvec{f}=\{\bvec{f}_v=\bvec{v}_v\}_{v\in\mathcal{V}}$. The state is therefore a point cloud $\{(\bvec{x}_v,\bvec{f}_v)\}_{v\in\mathcal{V}}$, which satisfies the requirement \ref{item:state_point_cloud} of ${\rm E}(3)$-symmetric MGs.

\textbf{Observation} The agent $i$'s local observation consists of the relative positions $\bvec{x}^i=\{\bvec{x}_i^i=\bvec{x}_i\}\cup\{\bvec{x}^i_v=\bvec{x}_v-\bvec{x}_i\}_{v\in \mathcal{N}\setminus \{i\}}$ of the set of entities $\mathcal{N}$, with the associated (relative) features $\bvec{f}^i=\{\bvec{f}_i^i=\bvec{v}_i\}\cup\{\bvec{f}_v^i=\bvec{0}\}_{v\in\mathcal{V}\setminus \{i\}}$. (Other agents' velocities are not observable). The agent $i$'s local observation is therefore a point cloud $\{(\bvec{x}_v^i,\bvec{f}_v^i)\}_{v\in\mathcal{V}}$ from agent $i$'s perspective, which satisfies the requirement \ref{item:observation_point_cloud} of ${\rm E}(3)$-symmetric MGs.

\textbf{Action} The local action spaces are absolute velocities, and are therefore Euclidean spaces which satisfies the requirement of \ref{item:action_space}. Denote the set of actions for the controllable agents as $A^a=\{\bvec{a}_i\}_{i\in \mathcal{N}^a}$.

The details of each scenarios in MPE is as the following. 

Cooperative Navigation: The collective goal of agents is to cover all the landmarks. There are two versions with $N^l=3,6$ landmarks for $N= 3, 6$ agents, respectively. There are no preys so $N^p=0$. 

Cooperative Push:
The collective goal of agents is to push a large ball to a target position. There are two versions with $N^l=2,2$ landmarks for $N= 3, 6$ agents, respectively. There are no preys so $N^p=0$.

Predator-and-Prey:
The collective goal of slowly moving agents (predators) is to capture some fast moving preys. The preys are pre-trained and controlled by the environment. There are $N^p=1,2$ preys and $N^l=2,3$ landmarks (blocks) for $N = 3, 6$ agents (predators), respectively.

\subsection{MuJoCo Continuous Control Tasks}
\label{sec:DMCS descriptions}

We choose several representative tasks with different levels of symmetries from the MuJoCo continuous control tasks, including the 2D ones from \cite{tassa2018deepmind} and 3D ones from \cite{pmlr-v202-chen23i} with single- and multi-agent variations: cartpole (balance, sparse), single- and multi-agent reacher (hard), single- and multi-agent swimmer with two links, and multi-agent variants of 3D hopper and walker. The details of the state, observation, and action in the selected scenarios are as the following. 

Cartpole (balance, sparse): the goal of cartpole is to balance a pole connecting to cart moving in the $x$ axis. It has rotation-equivariancy of  $\rm{rot}_{180^\circ}$ with respect to the $z$ axis and no translation-invariancy along the $x$ axis because the agent is tasked to balance the cart and the pole around the origin ($x=0$), not around some arbitrary position. The state contains the position of the cart, pole, and the origin, and the velocity of the cart and the pole. This task is fully observable, so the observation is the same as the state. The action is the 1D force acted on the cart on the $x$ (horizontal) axis. 

Single-agent reacher (hard): the goal of single-agent reacher is to control a robot with two links to reach a target. The second link (root, arm) can move around a hinge fixed at the root (origin), and the first link (hand, finger) can move around a hinge at the end of the second link. This task has rotation-equivariancy in the $x,y$ plane, and no translation-invariancy due to the hinge fixed at the origin. The state contains the positions of the target, finger, hand, arm, and root, and the velocities of the finger, hand, arm, and root. This task is fully observable, so the observation is the same as the state. The action is the torques applied to the two links. The target is generated randomly between circles whose centers are located at the origin with radius 0.05 and 0.2, respectively. 

Multi-agent reacher: the multi-agent reacher has the same goal as the single-agent reacher. The state is the same as the single-agent reacher. This task is partially observable. Agent 1 controlling the the first link can observe its own id, and the positions of the target, finger, and hand, and the velocities of the finger and hand. Agent 2 controlling the second link can observe its own id, and the positions of the target, arm, and root, and the velocities of the arm and root. The actions are the torques applied on the two links for agent 1 and agent 2, respectively.

Single-agent swimmer: the goal of single-agent swimmer is to control a robot with two controllable links to move to a target. It is similar to the single-agent reacher, except that the second link is not fixed at the origin. This task has both rotation- and translation-invariancy in the $xy$ plane. The state is the positions of the target, nose (in the front of the swimmer's head), and the first joint, the velocities and the angular velocities of the swimmer's bodies, and the joint angles. This task is fully observable, so the observation is the same as the state. The actions are the torques applied on the two controllable links. The target is generated randomly inside a square box whose center is located at the origin with width=height=0.3.

Multi-agent swimmer: the multi-agent swimmer has the same goal as the single-agent swimmer. The state is the same as the single-agent swimmer. We split the robot into two agents, with agent 1 controlling the first link, and agent 2 controlling the second link. This task is fully observable, so the observation is the same as the state, except that each agent can also observe their own ids. The actions are the torques applied on the two links for agent 1 and agent 2, respectively. 

Multi-agent hopper 3D (3\_shin): the goal of multi-agent hopper is to control a robot with 3 body parts (torso, thigh, foot) to reach a target position. This task has both rotation- and translation-invariancy in the $xy$ plane. The state contains the positions of the target, torso, thigh, and foot, the velocities of the torso, thigh, foot, the rotation axes of the thigh and foot, and the gravity which is a constant (0,0,-9.81). Agent 1 control the thigh, and agent 2 control the foot. This task is fully observable, so the local observation is the same as the state, except that each agent can also observe their own ids. The actions are the torques applied on thigh for agent 1 and the torques applied on foot for agent 2, respectively.

Multi-agent walker 3D (3\_left\_leg\_right\_foot): the goal of multi-agent walker is to control a robot with 5 body parts (torso, right thigh, right shin, left thigh, left shin) to reach a target position. This task has both rotation- and translation-invariancy in the $xy$ plane. The state contains the positions of the target, torso, right thigh, right shin, left thigh, and left shin, the velocities of the torso, right thigh, right shin, left thigh, and left shin, the rotation axes of the right thigh, right shin, left thigh, and left shin, and the gravity which is a constant (0,0,-9.81). Agent 1 control the right thigh, and agent 2 control the right shin. This task is fully observable, so the local observation is the same as the state, except that each agent can also observe their own ids. The actions are the torques applied on the right thigh for agent 1 and the torques applied on right shin for agent 2, respectively.

\subsection{SMAC}
The StarCraft Multi-Agent Challenge (SMAC) \cite{samvelyan2019starcraft} has become one of the most popular MARL benchmarks. All scenarios in SMAC can be described by ${\rm E}(3)$-symmetric Markov games in a similar way to MPE. We choose the \textit{Hard} scenario 8m\_vs\_9m and \textit{Super Hard} scenario 6h\_vs\_8z to evaluate our proposed algorithm, which has 8 agents and 6 agents, respectively. 

\textbf{Notation} We denote the set of controllable ally agents and uncontrollable enemies, which are basic units in scenarios of SMAC, as $\mathcal{N}^a=\{N+1,...,N^a\}$ and $\mathcal{N}^{e}=\{N^a+1,...,N^a+N^e\}$, where ${N}^a$ and ${N^e}$ are the number of agents and enemies, respectively. Here only agents are controllable and so $N^a=N$.
We have the set of nodes $\mathcal{V}=\mathcal{N}^a\cup\mathcal{N}^e$. The absolute positions associated with $\mathcal{N}^a$ and $\mathcal{N}^{e}$ are $\{\bvec{x}_i\}_{i\in \mathcal{N}^a}$ and $\{\bvec{x}_v\}_{v\in \mathcal{N}^e}$, respectively.

The local action spaces of the controllable ally agents are categorical, which is the union of the set of move directions, attack enemies, stop and no operation, i.e.,  $\mathcal{A}^i=\{\text{north, south, east, west}\}\cup \{\text{attack[enemy\_id]}\}_{\text{enemy\_id}\in \mathcal{N}^{e}} \cup \{\text{stop, no-op}\}$. The health associated with $\mathcal{N}^a$ and $\mathcal{N}^e$ are $H^{a}=\{\text{h}_i\}_{i\in \mathcal{N}}$ and $H^{e}=\{\text{h}_v\}_{v\in \mathcal{N}^{e}}$, respectively. The cooldown of the weapons associated with $\mathcal{N}^a$ and $\mathcal{N}^e$ which is not applicable are $CD^{a}=\{\text{cd}_i\}_{i\in \mathcal{N}}$ and $CD^{e}=\{\textbf{0}\}_{v\in \mathcal{N}^e}$, respectively. The shield associated with $\mathcal{N}^a$ and $\mathcal{N}^e$ are $SH^{a}=\{\text{sh}_i\}_{i\in \mathcal{N}}$ and $SH^e=\{\text{sh}_v\}_{v\in \mathcal{N}^e}$, respectively. Let $H=H^a\cup H^e$, $CD=CD^a\cup CD^e$, $SH=SH^a\cup SH^e$,  which are the sets of health, cooldown, and shield of the allies' and enemies' agents associated with the set of nodes $\mathcal{V}$.

\textbf{State} The state consists of the positions $\bvec{x}$ of the set of entities $\mathcal{N}$, with the associated features $\bvec{f}=\{\bvec{f}_v=(\text{h}_v,\text{cd}_v,\text{sh}_v)\}_{v\in\mathcal{V}}$. The state is therefore a point cloud $\{(\bvec{x}_v,\bvec{f}_v)\}_{v\in\mathcal{V}}$, which satisfies the requirement \ref{item:state_point_cloud} of ${\rm E}(3)$-symmetric MGs.

\textbf{Observation} Agent $i$'s observation $o^i(s)$ only contains information of itself and some other observable entities, denoted as $\mathcal{V}^i
\subseteq \mathcal{V}$. The information of the non-observable entities are padded as $\bvec{0}$. From agent $i$'s perspective, denote the visibility of the entities as $VI^i=\{\text{vi}_v=1\}_{v\in \mathcal{V}^i}\cup\{\text{vi}_v=0\}_{v\in \mathcal{V}\setminus\mathcal{V}^i}$, and denote the set of (one-hot) move directions $\in\{\text{north, south, east, west}\}$ as $MD^i=\{\text{md}_i\}\cup\{\bvec{0}\}_{v\in \mathcal{V}\setminus\{i\}}$. (only its own move direction is observable). Then, $o^i(s)$ consists of the relative positions $\bvec{x}^i=\{\bvec{x}_i^i=\bvec{0}\}\cup\{\bvec{x}^i_v=\bvec{x}_v-\bvec{x}_i\}_{v\in \mathcal{N}\setminus \{i\}}$ of the set of entities $\mathcal{N}$, with the associated (relative) features $\bvec{f}^i=\{f_i^i=(\text{agent\_id}=i,\text{md}_i,\text{vi}_i=1,\|\bvec{x}_i^i\|=0,\bvec{x}_i^i=\bvec{0},\text{h}_i,\text{sh}_i\}\cup\{f_v^i=(\text{agent\_id}=0,\text{md}_v=\bvec{0},\text{vi}_v,\|\bvec{x}_v^i\|,\bvec{x}_v^i,\text{h}_v,\text{sh}_v)\}_{v\in\mathcal{V}\setminus \{i\}}$. (Other entities' ids are not observable, so padded as 0). The agent $i$'s local observation is therefore a point cloud $\{(\bvec{x}_v^i,\bvec{f}_v^i)\}_{v\in\mathcal{V}}$ from agent $i$'s perspective, which satisfies the requirement \ref{item:observation_point_cloud} of ${\rm E}(3)$-symmetric MGs.

\textbf{Action} Each agent $i$'s action space is categorical, which is the union of the sets of move directions, attack enemies, stop and no operation, i.e., $\mathcal{A}^i=\{\text{north, south, east, west}\}\cup \{\text{attack[enemy\_id]}\}_{\text{enemy\_id}\in \mathcal{N}^{e}} \cup \{\text{stop, no-op}\}$. Obviously, $\mathcal{A}^i$ is not a Euclidean space and therefore does not satisfy the requirement of \ref{item:action_space}.

The details of the two scenarios in SMAC we test our algorithm is as the following.

8m\_vs\_9m: 8 controllable ally agents of type Marines are fighting against 9 uncontrollable enemies of type Marines. Shield is not applicable for Marines so both ${SH}^a$ and ${SH}^e$ are $\emptyset$.

6h\_vs\_8z: 6 controllable ally agents of type Hydralisks are fighting against 9 uncontrollable enemies of type Zealots. Shield is not applicable for Hydralisks so ${SH}^a=\emptyset$.

\section{Implementation Details}
\label{sec:Implementation details}
As described in \ref{subsec:E3-equivariant message passing}, the input to SEGNN is an input Euclidean graph is represented as a tuple $G^{\rm in}=(\mathcal{V}, \mathcal{E}, \bvec{x}, \bvec{f}^{\rm in})$ where
$\mathcal{V}$ is a set of vertices with 3D positions $\bvec{x}_v\in\R^3$, 
$\mathcal{E}\subseteq \mathcal{V} \times \mathcal{V}$ is a set of edges,
and $\bvec{f}^{\rm in}$ is a set of feature vectors, each associated with a vertex or an edge. Precisely, the input feature $\bvec{f}^{\rm in}$ can be further decomposed into node feature, node attribute, and edge attribute, i.e., $\bvec{f}^{\rm in}=(\bvec{f}^{\rm node}_{\rm feature},\bvec{f}^{\rm node}_{\rm attribute},\bvec{f}^{\rm edge}_{\rm attribute})$. Note that in MPE and SMAC, entities are in a 2D world, so we pad the dummy $z$ coordinate of $0$ for them to be in a 3D space whenever applicable, e.g. velocities, positions, etc.

Below we describe in details, whichever applicable, how to represent state $s$, state-action pair $(s,a)$, and observation $o^i(s)$ as Euclidean graphs for centralized state-value critic $V(s)$, action-value critic $Q(s,a)$, and observation-based actor $\nu^i(o^i(s))$, respectively, in MPE, MuJoCo tasks, and SMAC.

\subsection{MPE}
In MPE, we build our architecture based on MADDPG, where the critic is based on 
state and action pair and therefore we construct state-action based Euclidean graph processed by ${\rm E}(3)$-invariant critic. We also construct an observation-based Euclidean graph processed by an ${\rm E}(3)$-equivariant actor. 

\textbf{State-action based Euclidean graph}
Based on the information in the state $s=\{(\bvec{x}_v,\bvec{f}_v)\}_{v\in\mathcal{V}}$, $\mathcal{V}$ and $\bvec{x}$ have already been well-defined. The graph $\mathcal{E}$ is densely connected, i.e., $\mathcal{E}=\{(v_1,v_2):v_1\in \mathcal{V},v_2\in \mathcal{V},v_1\neq v_2\}$. We can pad the dummy $\bvec{0}$ action for non-controllable entities, i.e., landmarks and preys, to have actions for all entities: $A=\{\bvec{a}_i\}_{i\in \mathcal{N}^a}\cup\{\bvec{0}\}_{v\in \mathcal{N}^l}\cup\{\bvec{0}\}_{v\in \mathcal{N}^p}$. We can then design the node feature $\bvec{f}^{\rm node}_{\rm feature}$, node attribute $\bvec{f}^{\rm node}_{\rm attribute}$, and edge attribute $\bvec{f}^{\rm edge}_{\rm attribute}$ based on the information in the input features $\bvec{f}^{\rm in}=\bvec{f}=\{\bvec{v}_v\}_{v\in\mathcal{V}}$ and actions $A$. The node feature of each node consists of its own absolute velocity, the $l_2$ norm of its own absolute velocity, its own action, and $l_2$ norm of its own action. Specifically, $\bvec{f}^{\rm node}_{\rm feature}=\{[\bvec{v}_v,\norm{\bvec{v}_v}, \bvec{a}_v,\norm{\bvec{a}_v}]\}_{v\in\mathcal{V}}$. We can specify the ${\rm E}(3)$-equivariancy of both absolute velocities and actions, and ${\rm E}(3)$-invariancy of the associated $l_2$ norms in SEGNN to preserve physical constraints.  The node attribute is the set of node\_types $\in\{\text{agent, landmark, prey}\}$, i.e., $\bvec{f}^{\rm node}_{\rm attribute}=\{[\text{node\_type}(v)]\}_{v\in\mathcal{V}}$. The node types are ${\rm E}(3)$-invariant. The edge attribute is the set of relative positions, i.e., $\bvec{f}^{\rm edge}_{\rm attribute}=\{[\bvec{x}_{v_1}-\bvec{x}_{v_2}]\}_{(v_1,v_2)\in\mathcal{E}}$, which is ${\rm E}(3)$-equivariant. 

\begin{figure}[t]
\begin{center}
\includegraphics[width=\textwidth]{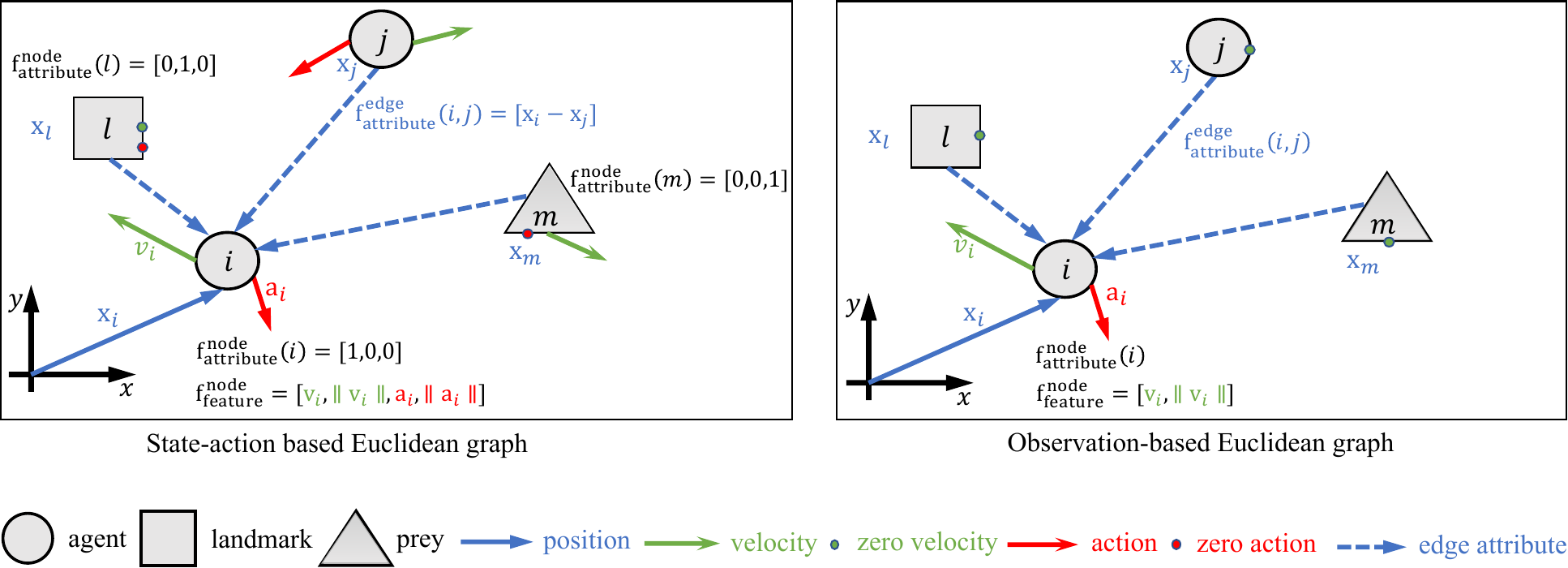}
\caption{Illustration of state-action based point cloud and observation-based point cloud for a scenario in MPE with agent, landmark and prey.}
\label{fig:Illustration of E3-symmetric markov games}
\end{center}
\end{figure}

\textbf{Observation-based Euclidean graph}
The relative positions in the observation can be restored to the absolute location, i.e., $\bvec{x}^i=\{\bvec{x}_i^i=\bvec{x}_i\}\cup\{\bvec{x}_v-\bvec{x}_i\}_{v\in \mathcal{N}\setminus \{i\}}\longrightarrow \{\bvec{x}_i^i=\bvec{x}_i\}\cup\{\bvec{x}_v-\bvec{x}_i+\bvec{x}_i\}_{v\in \mathcal{N}\setminus \{i\}}$. Then we can generate the observation-based point cloud in almost the same way, with two differences. The first difference is that due to partial observability, the velocities of other agents and preys are not available, so we have $V^a=\{\bvec{v}_i\}\cup \{\bvec{0}\}_{j\in \mathcal{N}\setminus \{i\}}$ and $V^p=\{\bvec{0}\}_{\mathcal{N}^p}$. Another difference is that the observation-based point cloud does not condition on actions, so correspondingly the node features do not contain action related elements and is only velocity related, i.e., $\bvec{f}^{\rm node}_{\rm feature}=\{[\bvec{v}_v,\norm{\bvec{v}_v}]\}_{v\in\mathcal{V}}$. 
\subsection{MuJoCo Continuous Control Tasks}
\label{sec:DMCS point cloud details}
In MuJoCo continuous control tasks, we build our architecture based on MADDPG, where the critic is based on 
state and action pair and therefore we construct state-action based Euclidean graph processed by ${\rm E}(3)$-invariant critic. We also construct an observation-based Euclidean graph processed by an ${\rm E}(3)$-equivariant actor. Note that the input to both our architecture and the MLP-based baselines contains the same information. Below we specify the state-action based Euclidean graph and the observation-based Euclidean graph for all the selected tasks. For all the point sets defined below, the graph $\mathcal{E}$ is densely connected, i.e., $\mathcal{E}=\{(v_1,v_2):v_1\in \mathcal{V},v_2\in \mathcal{V},v_1\neq v_2\}$. The node attribute is the set of node\_types, where each point in the point set has its own node type in all the selected MuJoCo tasks. The node types are ${\rm E}(3)$-invariant. The edge attribute is the set of relative positions, i.e., $\bvec{f}^{\rm edge}_{\rm attribute}=\{[\bvec{x}_{v_1}-\bvec{x}_{v_2}]\}_{(v_1,v_2)\in\mathcal{E}}$, which is ${\rm E}(3)$-equivariant.

\textbf{cartpole (Balance Sparse)} State-action based Euclidean graph: the point set $\mathcal{V}=\{\rm origin, cart\}$. For $\rm point_{cart}$, the node feature $\bvec{f}^{\rm node, cart}_{\rm feature}=\{[\bvec{v}_j,\norm{\bvec{v}_j}]\}_{j\in \{\rm cart, pole\}}\cup \{[\bvec{x}_{\rm pole},\norm{\bvec{x}_{\rm pole}}]\}\cup \{\bvec{a},\norm{\bvec{a}}\}$. For $\rm point_{origin}$, the node feature $\bvec{f}^{\rm node, origin}_{\rm feature}=\{\bvec{0}\}$. The node attribute is the node\_type: $\forall i\in\mathcal{V},\bvec{f}^{\rm node}_{\rm attribute}=[\text{node\_type}(i)]$. Observation based Euclidean graph: it is the same as the state-action based Euclidean graph, except that action is not included in the node feature.

\textbf{Single-agent reacher (hard)} State-action based Euclidean graph: the point set $\mathcal{V}=\{\rm target, finger\}$. For $\rm point_{finger}$, the node feature $\bvec{f}^{\rm node, finger}_{\rm feature}=\{[\bvec{v}_j,\norm{\bvec{v}_j}]\}_{j\in \{\rm finger, hand, arm, root\}}\cup \{[\bvec{x}_j,\norm{\bvec{x}_j}]\}_{j\in \{\rm hand, arm, root\}}\}\cup \{\bvec{a},\norm{\bvec{a}}\}$. For $\rm point_{target}$, the node feature $\bvec{f}^{\rm node, target}_{\rm feature}=\{\bvec{0}\}$. The node attribute is the node\_type: $\forall i\in\mathcal{V},\bvec{f}^{\rm node}_{\rm attribute}=[\text{node\_type}(i)]$. Observation based Euclidean graph: it is the same as the state-action based Euclidean graph, except that action is not included in the node feature.

\textbf{Multi-agent reacher (hard)} The state-action based Euclidean graph is the same as the single-agent reacher. For agent 1 controlling the green part in reacher in Figure \ref{fig:MuJoCo_results_illustration}, the point set is $\mathcal{V}=\{\rm target, finger\}$. The node feature $\bvec{f}^{\rm node, finger}_{\rm feature}=\{\rm id_{agent\_1}\}\cup\{[\bvec{v}_j,\norm{\bvec{v}_j}]\}_{j\in \{\rm finger, hand\}}\cup \{[\bvec{x}_j,\norm{\bvec{x}_{\rm hand}}]\}\}\cup \{\bvec{a}_{1},\norm{\bvec{a}_{1}}\}$. For $\rm point_{target}$, the node feature $\bvec{f}^{\rm node, target}_{\rm feature}=\{\bvec{0}\}$. For agent 2 controlling the yellow part in reacher in Figure \ref{fig:MuJoCo_results_illustration}, the point set is $v=\{\rm target, arm\}$. The node feature $\bvec{f}^{\rm node, arm}_{\rm feature}=\{\rm id_{agent\_2}\}\cup\{[\bvec{v}_j,\norm{\bvec{v}_j}]\}_{j\in \{\rm arm, root\}}\cup \{[\bvec{x}_j,\norm{\bvec{x}_{\rm root}}]\}\}\cup \{\bvec{a}_{2},\norm{\bvec{a}_{2}}\}$. For $\rm point_{target}$, the node feature $\bvec{f}^{\rm node, target}_{\rm feature}=\{\bvec{0}\}$. 

\textbf{Single-agent swimmer} State-action based Euclidean graph: the point set $\mathcal{V}=\{\rm target, nose, first\_joint\}$. For $\rm point_{nose}$, the node feature $\bvec{f}^{\rm node, nose}_{\rm feature}=\{[\bvec{v}_j,\norm{\bvec{v}_j}, \bvec{\omega}_j,\norm{\bvec{\omega}_j}]\}_{j\in \{\rm bodies\}}\cup \{[\theta_j]\}_{j\in \{\rm joints\}}\cup \{\bvec{a},\norm{\bvec{a}}\}$. For $\rm point_{target}$ and $\rm point_{first\_link}$, the node features are $\{\bvec{0}\}$. The node attribute is the node\_type: $\forall i\in\mathcal{V},\bvec{f}^{\rm node}_{\rm attribute}=[\text{node\_type}(i)]$. Observation based Euclidean graph: it is the same as the state-action based Euclidean graph, except that action is not included in the node feature.

\textbf{Multi-agent swimmer} The state-action based Euclidean graph is the same as the single-agent swimmer. The observation-based Euclidean graphs for both agents are the same as the state-action based Euclidean graph, except that action is not included in the node feature and their own ids are included in the node feature which are ${\rm E}(3)$-invariant.

\textbf{Multi-agent hopper} State-action based Euclidean graph: the point set $\mathcal{V}=\{\rm target, torso, thigh, foot\}$. For the for noncontrollable components target and torso, we set the dummy action of $\{\bvec{0}\}$. For $i \in \mathcal{V}$, the node feature $\bvec{f}^{\rm node, i}_{\rm feature}=[\bvec{v}_i,\norm{\bvec{v}_i}, \bvec{x}_i,\norm{\bvec{x}_i},\rm ,\rm rotation\_axis\_i,\norm{rotation\_axis\_i},\bvec{a}_i,\norm{\bvec{a}_i}]$. The node attribute is the node\_type: $\forall i\in\mathcal{V},\bvec{f}^{\rm node}_{\rm attribute}=[\text{node\_type}(i)]$. Observation based Euclidean graph: it is the same as the state-action based Euclidean graph, except that action is not included in the node feature and their own ids are included in the node feature which are ${\rm E}(3)$-invariant.

\textbf{Multi-agent walker} State-action based Euclidean graph: the point set $\mathcal{V}=$ \{\rm target, torso, right thigh, right shin, left thigh, left shin\}. For the for noncontrollable components \{target, torso, left thigh, left shin\}, we set the dummy action of $\{\bvec{0}\}$. For $i \in \mathcal{V}$, the node feature $\bvec{f}^{\rm node, i}_{\rm feature}=[\bvec{v}_i,\norm{\bvec{v}_i}, \bvec{x}_i,\norm{\bvec{x}_i},\rm ,\rm rotation\_axis\_i,\norm{rotation\_axis\_i},\bvec{a}_i,\norm{\bvec{a}_i}]$. The node attribute is the node\_type: $\forall i\in\mathcal{V},\bvec{f}^{\rm node}_{\rm attribute}=[\text{node\_type}(i)]$. Observation based Euclidean graph: it is the same as the state-action based Euclidean graph, except that action is not included in the node feature and their own ids are included in the node feature which are ${\rm E}(3)$-invariant.

\subsection{SMAC}
In SMAC, we build our architecture based on MAPPO, where the critic is based on state only and therefore we construct state based Euclidean graph processed by ${\rm E}(3)$-invariant critic. The action space in SMAC is categorical and therefore we use a traditional MLP-based actor to process the observation.

\textbf{State based Euclidean graph} 
Based on the information in the state $s=\{(\bvec{x}_v,\bvec{f}_v)\}_{v\in\mathcal{V}}$, $\mathcal{V}$ and $\bvec{x}$ have already been well-defined. The graph $\mathcal{E}$ is a nearest neighbor graph degree $k$, i.e., $\mathcal{E}=\{(v_1,v_2):v_1\in \mathcal{V},v_2\in \mathcal{V},v_1\in \text{neighbor}(v_2,k)\}$. We can then design the node feature $\bvec{f}^{\rm node}_{\rm feature}$, node attribute $\bvec{f}^{\rm node}_{\rm attribute}$, and edge attribute $\bvec{f}^{\rm edge}_{\rm attribute}$ based on the information in the input features $\bvec{f}^{\rm in}=\bvec{f}=\{\bvec{v}_v\}_{v\in\mathcal{V}}$. The node feature of each node consists of its health, cooldown, and shield. Specifically, $\bvec{f}^{\rm node}_{\rm feature}=\{[\text{h}_v,\text{cd}_v,\text{sh}_v]\}_{v\in\mathcal{V}}$. All quantities in node features are ${\rm E}(3)$-invariant. The node attribute is the set of node\_types $\in\{\text{[Team,Type]}\}_{\text{Team}\in\{\text{ally, enemy}\},\text{Type}\in\{\text{Marines, Hydralisks, Zealots}\}}$, i.e., $\bvec{f}^{\rm node}_{\rm attribute}=\{[\text{node\_type}(v)]\}_{v\in\mathcal{V}}$. The node types are ${\rm E}(3)$-invariant. The edge attribute is the set of relative positions, i.e., $\bvec{f}^{\rm edge}_{\rm attribute}=\{[\bvec{x}_{v_1}-\bvec{x}_{v_2}]\}_{(v_1,v_2)\in\mathcal{E}}$, which is ${\rm E}(3)$-equivariant.

\newpage\clearpage
\section{Supplementary Results}
\label{sec:Supplementary results}
\subsection{Invariancy Measure in MPE}
\label{sec:invariancy measurements in MPE}
In MPE, the metrics for measuring invariancy for actor and critic in terms of rotation and translation are defined in the following:

\begin{align}
\label{criti-rotation}
\textstyle
    {\rm{Invariancy_{Q}^{rot}}}=-\frac{1}{|A|}\sum_{\alpha\in A}
    |Q(s,a)-Q({\rm rot}_\alpha[s],{\rm rot}_\alpha[a])|
\end{align}

\begin{align}
\label{critic-translation}
\textstyle
    {\rm{Invariancy_{Q}^{transl}}}=-\frac{1}{|L|}\sum_{l\in L}
    |Q(s,a)-Q({\rm transl}_l[s],{\rm transl}_l[a])|
\end{align}

\begin{align}
\label{actor-rotation}
\textstyle
    {\rm{Invariancy_{\nu}^{rot}}} = \frac{1}{|A|N}\sum_{i\in\mathcal{N}}\sum_{\alpha\in A}
    \cos\left(
    {\rm rot}_\alpha[\nu^i(o^i(s)],~
    \nu^i(o^i({\rm rot}_\alpha[s]))
    \right)
\end{align}

\begin{align}
\label{actor-translation}
\textstyle
    {\rm{Invariancy_{\nu}^{transl}}} = \frac{1}{|L|N}\sum_{i\in\mathcal{N}}\sum_{l\in L}
    \cos\left(
    {\rm transl}_l[\nu^i(o^i(s)],~
    \nu^i(o^i({\rm transl}_l[s]))
    \right)
\end{align}

,where ${\rm rot}_\alpha[\cdot]$ and ${\rm transl}_l[\cdot]$ performs the rotation by $\alpha$ and translation by $l$, respectively, and $\cos(\cdot,\cdot)$ measures the cosine similarity. $A$ is the list of angles which we use $[30^\circ,60^\circ\cdots,330^\circ]$, $L$ is the list of translations which we use $[(-l_x,0),(-0.5l_x,0),(0.5l_x,0),(l_x,0)]$ for translations along $x$ axis and $[(0,-l_y),(0,-0.5l_y),(0,0.5l_y),(0,l_y)]$ for translations along $y$ axis (the size of the map is $l_x$ by $l_y$, where $l_x=l_y=1,1.5,1$ for the scenarios of push, navigation, and prey, respectively). We calculate these four metrics by averaging over state and observation collected in different time steps in $200$ episodes. The ranges of those metrics are not important. By construction, the larger the values of ${\rm{Invariancy_{Q}^{rot}}},{\rm{Invariancy_{Q}^{transl}}},{\rm{Invariancy_{\nu}^{rot}}},$
and ${\rm{Invariancy_{\nu}^{transl}}}$ are, the closer the behaviors the non-SEGNN-based architecture are compared to those of the SEGNN-based architecture, which has the largest possible values for those metrics, in terms of their invariancy to the corresponding transformations.

\newpage\clearpage
\subsection{The Emergence of Invariancy in MPE.}
\label{sec:The emergence of invariancy and equivariancy in MPE}
As illustrated in Figure \ref{fig:MPE-Full-Emergence of equivariancy}, the GCN-based critic with permutation invariancy has better rotation- and translation-invariancy than the MLP-based one for the whole training process. One possible reason is that for certain configurations of the state, permutation invariancy is a special type of Euclidean invariancy, e.g., 6 agents whose positions are in a circle with $60^\circ$ between adjacent ones, which displays a rotation-invariancy of $\rm{rot}_{60^\circ}$. However, there is no increase in rotation- and translation-invariancy for both GCN- and MLP-based critics, which may explain why actors learned with SEGNN-based critics have better performance. On the other hand, the actor learned with a SEGNN-based critic emerges better rotation- and translation-invariancy than the MLP and GCN baselines. The actor learned with a GCN-based critic achieves similar rotation-invariancy but worse translation-invariancy than the one learned with an MLP-based critic. Note that the translation-invariancy for actors is initially high, because in the observation only the agent's own position is absolute and can be modified by translation, whereas all other values are relative and translation-invariant. The randomly initialized actors will therefore output similar actions. 
\begin{figure}[ht]
\begin{center}
\centerline{\includegraphics[width=0.77\columnwidth]{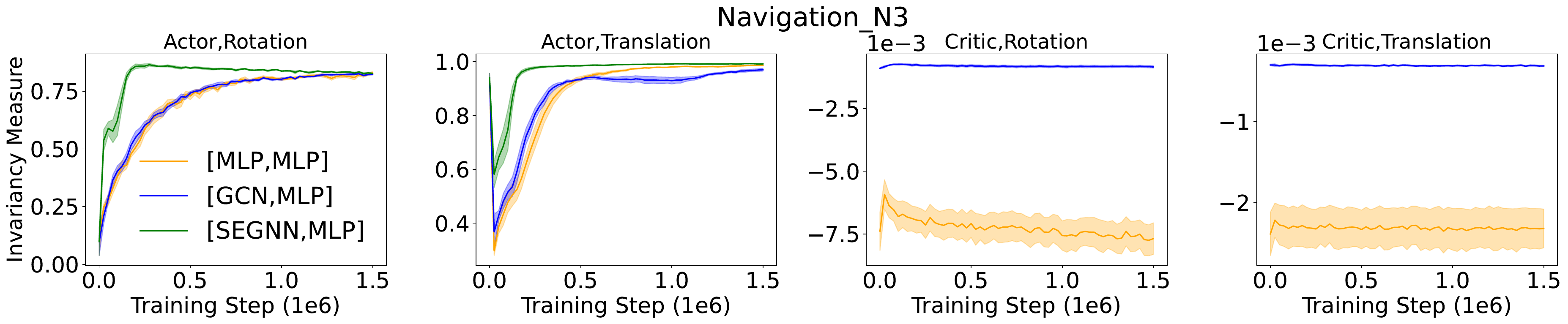}}
\centerline{\includegraphics[width=0.77\columnwidth]{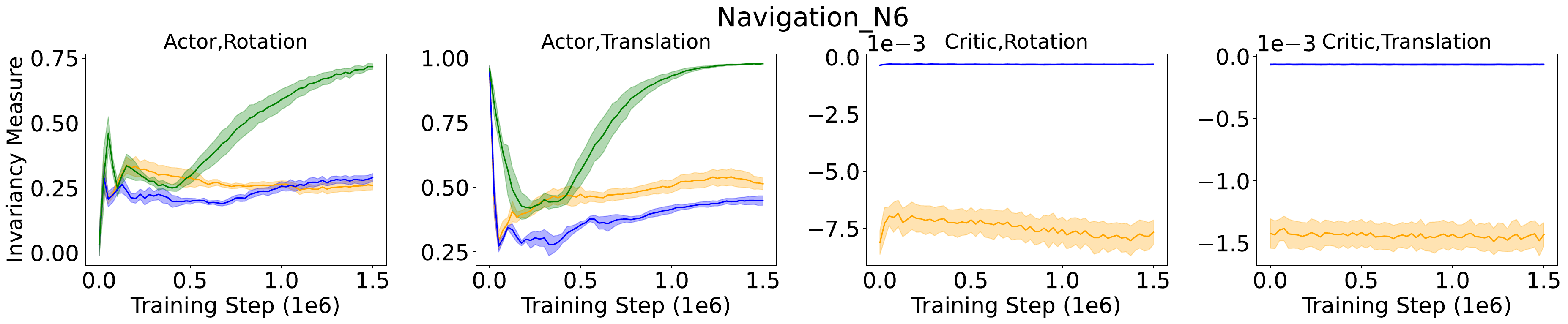}}
\centerline{\includegraphics[width=0.77\columnwidth]{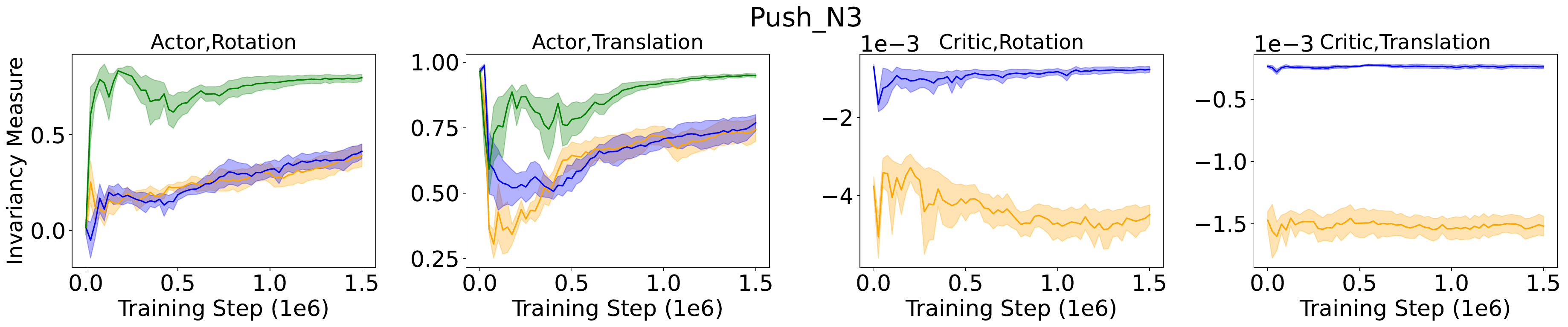}}
\centerline{\includegraphics[width=0.77\columnwidth]{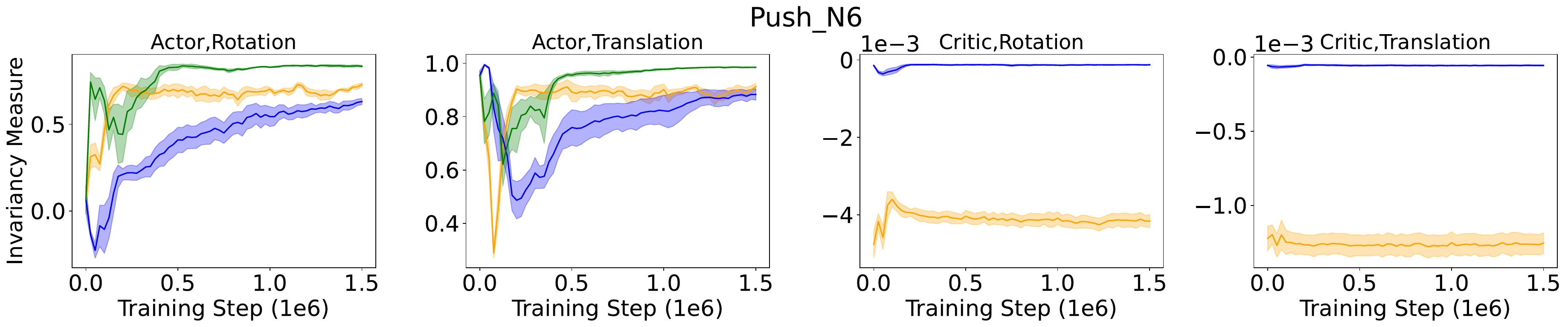}}
\centerline{\includegraphics[width=0.77\columnwidth]{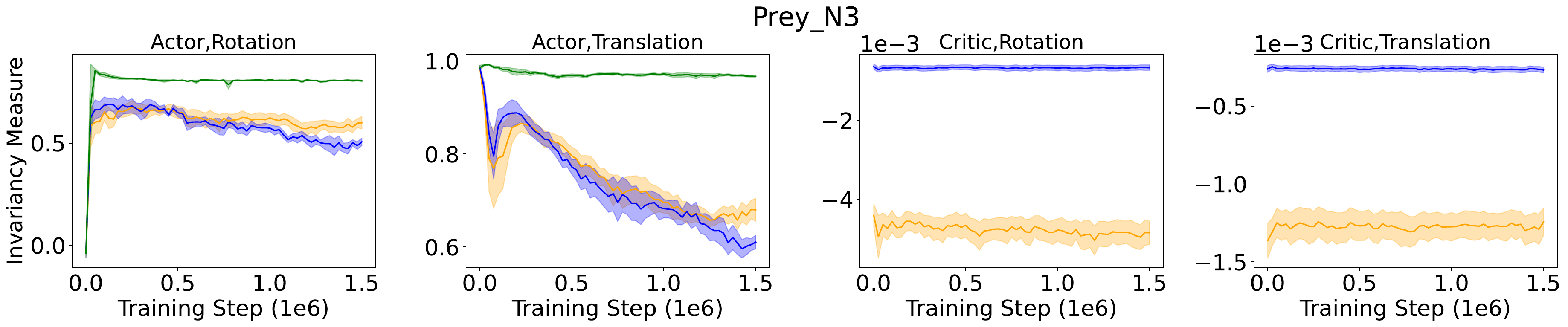}}
\centerline{\includegraphics[width=0.77\columnwidth]{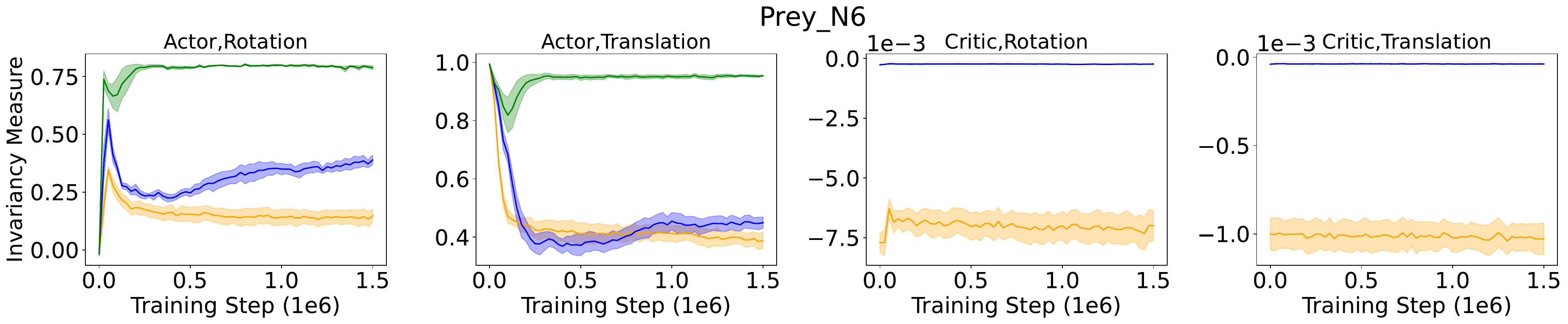}}
\caption{The emergence of invariancy and equivariancy in non-SEGNN-based critic and actor, respectively, in MPE. }
\label{fig:MPE-Full-Emergence of equivariancy}
\end{center}
\vskip -0.3in
\end{figure}

\newpage\clearpage
\subsection{Invariancy Measure in SMAC}
\label{sec:invariancy measurements in SMAC}
The critic used by MAPPO is a value function, so the invariancy for critic is defined in terms of $V$ instead of $Q$. The action of controllable ally agents in SMAC is categorical, so we use KL divergence instead of cosine similarity to measure the difference between two probability distributions. The action is the union of the set of move directions, attack enemies, stop and no operation, i.e.,  $A^a=\{\text{north, south, east, west}\}\cup \{\text{attack[enemy\_id]}\}_{\text{enemy\_id}\in \mathcal{N}^{e}} \cup \{\text{stop, no-op}\}$. Under rotation, only move directions will be changed. Its hard to measure how the move direction will change for an arbitrary rotation angle, so we only consider (counter clockwise) rotation angles of $[90^\circ,180^\circ,270^\circ]$, which will roll the move direction accordingly. For example, for a (counter clockwise) rotation angle of $90^\circ$, move to the north will change to move to the west. The translation for actor is already achieved in scenarios in SMAC due to the missing of absolute positions in the observation, so it is not included here. 

\begin{align}
\label{SMAC-critic-rotation}
\textstyle
    {\rm{Invariancy_{V}^{rot}}}=-\frac{1}{|A|}\sum_{\alpha\in A}
    |V(s)-V({\rm rot}_\alpha[s])|
\end{align}

\begin{align}
\label{SMAC-critic-translation}
\textstyle
    {\rm{Invariancy_{V}^{transl}}}=-\frac{1}{|L|}\sum_{l\in L}
    |V(s)-V({\rm transl}_l[s])|
\end{align}

\begin{align}
\label{SMAC-actor-rotation}
\textstyle
    {\rm{Invariancy_{\nu}^{rot}}} = \frac{1}{|A|N}\sum_{i\in\mathcal{N}}\sum_{\alpha\in A}
    {\rm KL}\left(
    {\rm rot}_\alpha[\nu^i(\cdot|o^i(s)],~
    \nu^i(\cdot|o^i({\rm rot}_\alpha[s]))
    \right)
\end{align}

, where ${\rm rot}_\alpha[\cdot]$ and ${\rm transl}_l[\cdot]$ performs the rotation by $\alpha$ and translation by $l$, respectively, and ${\rm KL}(\cdot,\cdot)$ measures the KL divergence. $A$ is the list of angles which we use $[90^\circ,180^\circ,270^\circ]$, $L$ is the list of translations which we use $[(-l_x,0),(-0.5l_x,0),(0.5l_x,0),(l_x,0)]$ for translations along $x$ axis and $[(0,-l_y),(0,-0.5l_y),(0,0.5l_y),(0,l_y)]$ for translations along $y$ axis (the normalized size of the map is $l_x$ by $l_y$, where $l_x=l_y=1$ for the scenarios in SMAC. We calculate these four metrics by averaging over state and observation collected in different time steps in $32$ episodes. The ranges of those metrics are not important. By construction, the larger the values of ${\rm{Invariancy_{V}^{rot}}},{\rm{Invariancy_{V}^{transl}}},$ and ${\rm{Invariancy_{\nu}^{rot}}}$
are, the closer the behaviors the non-SEGNN-based architecture are compared to those of the SEGNN-based architecture, which has the largest possible values for those metrics, in terms of their invariancy to the corresponding transformations.

\newpage\clearpage
\subsection{The Emergence of Invariancy and Equivariancy in SMAC.}
\label{sec:The emergence of invariancy and equivariancy in SMAC}
As illustrated in Figure \ref{fig:SMAC-Full-Emergence of equivariancy}, there is no obvious emergence of the rotation-invariancy for both actor and critic. This can be caused by the competitive nature of the scenarios in SMAC, where we also see a decrease in actors' translation-invariancy in competitive scenarios Prey\_N3 and Prey\_N6, as shown in the appendix. Another possible reason is that some of the categorical actions, the move directions \{north, south, east, west\}, are only rotation-equivariant to multiples of $90^\circ$ instead of arbitrary angles and therefore break the overall invariancy of the game.
\begin{figure}[ht]
\begin{center}
\centerline{\includegraphics[width=.79\columnwidth]{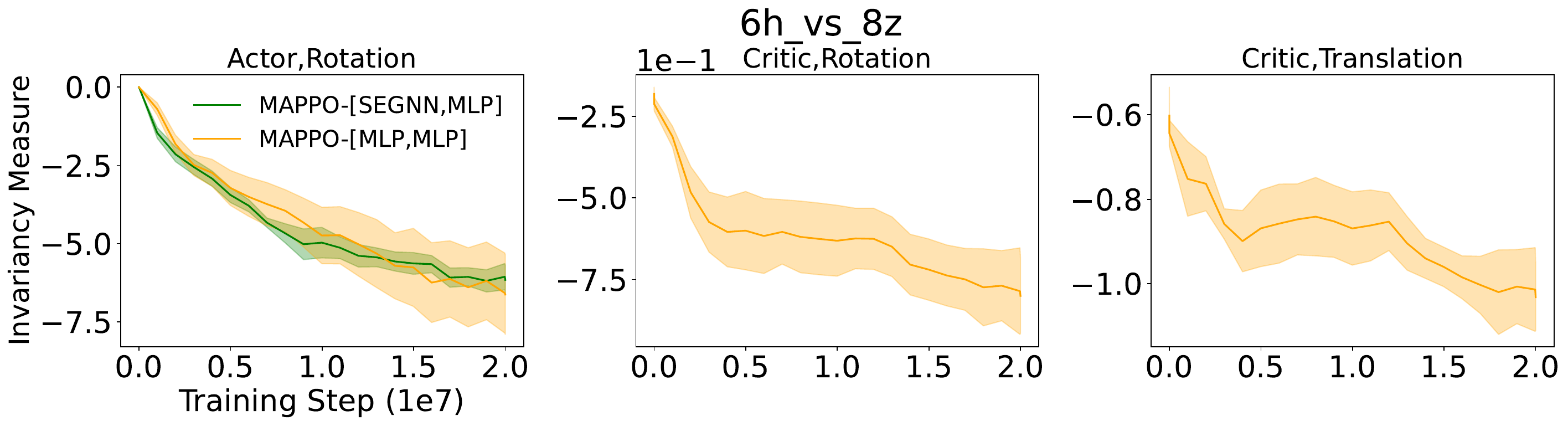}}
\centerline{\includegraphics[width=.79\columnwidth]{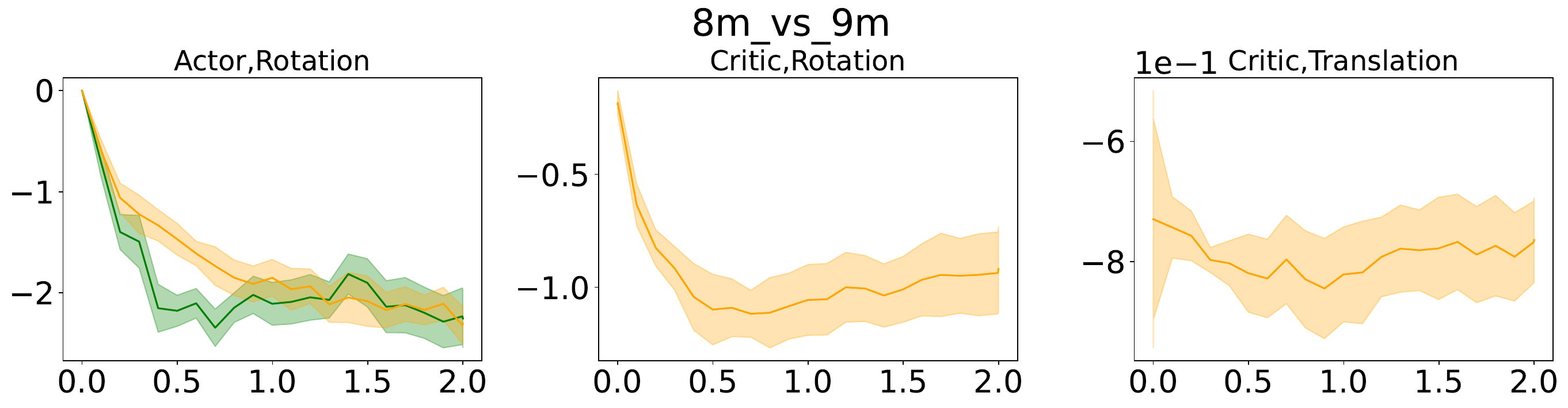}}
\caption{The emergence of invariancy and equivariancy in non-SEGNN-based critic and actor, respectively, in SMAC. }
\label{fig:SMAC-Full-Emergence of equivariancy}
\end{center}
\vskip -0.3in
\end{figure}

\newpage\clearpage
\section{Hyperparameters}
\label{sec:Hyperparameters}
\subsection{MPE}
\begin{table}[h]
\centering
\caption{Hyperparameters for MPE}
\begin{tabular}{ll} 
\toprule
Hyperparameter                       & Value                                                                           \\ 
\hline
Episode length                       & 25                                                                              \\
Number of training episodes          & 60000                                                                           \\
Discount factor                      & 0.95                                                                            \\
Batch size from replay buffer for [MLP, MLP] and [GCN, MLP]& 1024                                                                             \\
Batch size from replay buffer for [SEGNN, MLP] and [SEGNN, SEGNN]& 128                                                                             \\
Actor's learning rate for [MLP, MLP] and [GCN, MLP]& 1e-4                                                                           \\ 
Acror's learning rate for [SEGNN, MLP] in ordered scenarios& 3e-4,3e-4,1e-4,3e-4,3e-4,3e-4                                                                           \\ 
Acror's learning rate for [SEGNN, SEGNN] in ordered scenarios& 1e-4,3e-4,3e-5,3e-5,3e-4,3e-4                                                                           \\ 
Critic's learning rate for [MLP, MLP] and [GCN, MLP]& 1e-3                                                                           \\ 
Critic's learning rate for [SEGNN, MLP] in ordered scenarios& 1e-3,3e-4,1e-3,1e-3,1e-3,1e-3                                                                           \\ 
Critic's learning rate for [SEGNN, SEGNN] in ordered scenarios& 1e-3,1e-3,1e-3,1e-3,1e-3,1e-3                                                                           \\ 
Actor's and critic's learning rates for transfer learning in Navigation& 2e-4,1e-3                                                        \\
Actor's and critic's learning rates for transfer learning in Push& 3e-5,1e-3                                                        \\
Actor's and critic's learning rates for transfer learning in Prey& 3e-4,1e-3                                                        \\
Graph used by SEGNN for 6h\_vs\_8z & complete graph
\\
\#episodes per evaluation& 200                                                                           \\
\#seeds& 5                                                                     \\

\bottomrule
\end{tabular}
The ordered scenarios are: Navigation\_N3, Navigation\_N6, Push\_N3, Push\_N6, Tag\_N3, Tag\_N6.
\\Unless specified, the mentioned hyperparameter is applied to all scenarios in MPE.
\\Actors' learning rates are searched in [1e-4, 3e-5, 1e-5]. 
\\Critics' learning rates are searched in [1e-3, 3e-4, 1e-4]. 
\label{Hyperparameters for MPE}
\end{table}

\subsection{The MuJoCo Continuous Control Tasks}
\begin{table}[h!]
\centering
\caption{Hyperparameters for 2D tasks (cartpole, reacher, swimmer) in the MuJoCo tasks}
\begin{tabular}{ll} 
\toprule
Hyperparameter                       & Value                                                                           \\ 
\hline
Episode length                       & 1000                                                                              \\
Number of training steps          & 1.1e6                                                                          \\
Discount factor                      & 0.99                                                                            \\
Batch size                      & 256                                                                            \\

Acror's learning rate for [MLP, MLP] and [GCN, MLP] & 1e-4                                                                           \\ 
Acror's learning rate for [SEGNN, SEGNN]& 5e-5                                                                           \\ 
Critic's learning rate for [MLP, MLP] and [GCN, MLP] & 1e-4                                                                                                                                                \\ 
Critic's learning rate for [SEGNN, SEGNN]& 5e-5                                                                                                                                                \\ 
Graph used by SEGNN in all 2D tasks& complete graph
\\
\#episodes per evaluation for cartpole (balance, sparse)& 20                                                                           \\
\#episodes per evaluation for single- and multi-agent reacher& 10                                                                           \\
\#episodes per evaluation for single- and multi-agent swimmer& 100                                                                           \\

\#seeds for cartpole (balance, sparse), single-agent reacher, and single-agent swimmer& 5                                                                         \\
\#seeds for multi-agent reacher, and multi-agent swimmer& 10                                                                       \\

\bottomrule
\end{tabular}
\\Unless specified, the mentioned hyperparameter is applied to all the tasks and the algorithms.
\\The hyperparameters except SEGNN's learning rates are the default ones used by \cite{rezaei2022continuous}
\label{Hyperparameters for 2D MuJoCo tasks}
\end{table}

\begin{table}[ht!]
\centering
\caption{Hyperparameters for 3D tasks (hopper, walker) in the MuJoCo tasks}
\begin{tabular}{ll} 
\toprule
Hyperparameter                       & Value                                                                           \\ 
\hline
Episode length                       & 1000                                                                              \\
Number of training steps          & 1.1e6                                                                          \\
Discount factor                      & 0.99                                                                            \\
Batch size                      & 256                                                                            \\

Acror's learning rate & 3e-4                                                                           \\ 
Critic's learning rate& 3e-4                                                                                                                                                                                 \\ 
Graph used by SEGNN in all 3D tasks& complete graph
\\
\#episodes per evaluation & 10                                                                           \\

\#seeds & 5                                                                         \\

\bottomrule
\end{tabular}
\\Unless specified, the mentioned hyperparameter is applied to all the tasks and the algorithms.
\\The hyperparameters are the default ones used by \cite{pmlr-v202-chen23i}
\label{Hyperparameters for 3D MuJoCo tasks}
\end{table}

\newpage\clearpage
\subsection{SMAC}
\begin{table}[h]
\centering
\caption{Hyperparameters for SMAC}
\begin{tabular}{ll} 
\toprule
Hyperparameter                       & Value                                                                           \\ 
\hline
Episode length                       & 400                                                                              \\
Number of training steps          & 2e7                                                                          \\
Discount factor                      & 0.99                                                                            \\
\#Rollout threads& 8                                                                             \\
\#Training threads& 1
\\
PPO epoch for 8m\_vs\_9m, 6h\_vs\_8z             & 5,10                                          \\ 
\# mini-batch for 8m\_vs\_9m, 6h\_vs\_8z             & 1,4                                          \\ 
The degree of nearest neighbor graph used by SEGNN for 8m\_vs\_9m & 4
\\
The degree of nearest neighbor graph used by SEGNN for 6h\_vs\_8z & 13 (complete graph)
\\
Acror's learning rate for all algorithms& 5e-4                                                                           \\ 
Critic's learning rate for all algorithms& 5e-4                                                                                                                                                                        \\ 
\#episodes per evaluation& 32                                                                           \\
\#seeds& 5                                                                           \\

\bottomrule
\end{tabular}
\\Unless specified, the mentioned hyperparameter is applied to both 6h\_vs\_8z and 8m\_vs\_9m.
\\The hyperparameters are the default ones used by the original MAPPO \cite{yu2022surprising}
\label{Hyperparameters for SMAC}
\end{table}

\newpage\clearpage
\section{Number of Parameters in Neural Networks}
\label{sec:Number of parameters in neural networks}
\subsection{MPE}
\begin{tabular}{ |p{2.2cm}|p{1.9cm}|p{1.9cm}|p{1.9cm}|p{1.9cm}|p{1.9cm}|p{1.9cm}|  }
 \hline
 \multicolumn{7}{|c|}{MPE} \\
 \hline
 Alg/Env & Navigation\_N3&Navigation\_N6&Push\_N3&Push\_N6&Prey\_N3&Prey\_N6\\
 \hline
 MLP (critic)   & 22913    &38273&   22145&32129&23681&39809\\
 GCN (critic)&   37505    &40577&   36993&38529&38017&41089\\
 SEGNN (critic) &33791    &33791&   42344&42344&42344&42344\\
 MLP (actor)    &18690    &20226&   18434&19202&18946&20482\\
 SEGNN (actor)& 40553    &40553&   48532&48532&48532&48532\\
 \hline
\end{tabular}

\subsection{The MuJoCo 2D Continuous Control Tasks}
\begin{tabular}{ |p{2.2cm}|p{1.9cm}|p{1.9cm}|p{3.0cm}|p{1.9cm}|p{3.2cm}|}
 \hline
 \multicolumn{6}{|c|}{The MuJoCo 2D continuous control tasks} \\
 \hline
 Alg/Env & Cartpole&Reacher&Multi-agent Reacher&Swimmer&Multi-agent Swimmer\\
 \hline
 MLP (critic)   & 70401    &74497&   74497&71937&71937\\
 GCN (critic)&   37505    &40065&   40065&38529&38529\\
 SEGNN (critic) &33911    &34169&   34169&43193&43193\\
 MLP (actor)    &70145    &74242&   71425&71682&71937\\
 GCN (actor)&   36993    &39426&   38273&37890&38273\\
 SEGNN (actor)& 32694    &39234&   34003&48971&43124\\
 \hline
\end{tabular}

\subsection{The MuJoCo 3D Continuous Control Tasks}
\begin{tabular}{ |p{2.2cm}|p{1.9cm}|p{1.9cm}|}
 \hline
 \multicolumn{3}{|c|}{The MuJoCo 3D continuous control tasks} \\
 \hline
 Alg/Env & Hopper&Walker\\
 \hline
 MLP (critic)   & 81665    &90881\\
 GCN (critic)   & 72705    &73217\\
 SEGNN (critic) &51781    &69329\\
 MLP (actor)   & 80387    &88838\\
 GCN (actor)   & 72963    &74246\\
 SEGNN (actor) &67511    &134508\\
 \hline
\end{tabular}

\subsection{SMAC}
\begin{tabular}{ |p{2.2cm}|p{1.9cm}|p{1.9cm}|}
 \hline
 \multicolumn{3}{|c|}{SMAC} \\
 \hline
 Alg/Env & 6h\_vs\_8z&8m\_vs\_9m\\
 \hline
 MLP (critic)   & 12529    &12727\\
 SEGNN (critic) &33811    &33765\\
 \hline
\end{tabular}

\newpage\clearpage
\section{Computing Resources}
The code is implemented by PyTorch, and runs on NVIDIA Tesla V100 GPUs with 32 CPU cores. For MPE, a single run with [MLP, MLP], [GCN, MLP], [SEGNN, MLP], [SEGNN, SEGNN] takes approximately 2 hours, 3 hours, 7 hours, 10 hours to run, respectively. For the MuJoCo continuous control tasks, [MLP, MLP] takes approximately 3 hours to run, [GCN, MLP] takes approximately 5 hours to run, and [SEGNN, SEGNN] takes approximately 4 days to run. For SMAC, [MLP, MLP] and [SEGNN, MLP] takes approximately 8 hours and 4 days to run, respectively. 

\end{document}